\documentclass{article}
\usepackage[utf8]{inputenc}
\usepackage{authblk}
\usepackage{amsmath}
\usepackage{amsthm}
\usepackage{amsfonts}
\usepackage{thm-restate}
\usepackage{cleveref}
\usepackage{subcaption}
\usepackage{xcolor}
\usepackage{graphicx}
\usepackage{fullpage}
\usepackage[shortlabels]{enumitem}

\newtheorem{lemma}{Lemma}

\newtheorem{obs}{Observation}
\newtheorem{corollary}{Corollary}
\newtheorem{remark}{Remark}

\newtheorem*{prop*}{Properties}

\newtheorem*{basicprop*}{Basic Properties}

\newcommand{\bestparameterizedsolution}{O(k(|V| + |E|)\log{|V|})}
\newcommand{\tmp}{\mathcal{T}}
\newcommand{\pathcover}{\mathcal{P}}
\newcommand{\flowG}{\mathcal{G}}
\newcommand{\flowV}{\mathcal{V}}
\newcommand{\flowE}{\mathcal{E}}

\newcommand{\Nzero}{\mathbb{N}_{0}}

\newcommand{\invA}{\textbf{Invariant A}}
\newcommand{\invB}{\textbf{Invariant B}}
\newcommand{\invC}{\textbf{Invariant C}}

\newcommand{\residual}[2]{\mathcal{R}(#1, #2)}
\newcommand{\width}{\mathsf{width}}
\newcommand{\blue}{$blue$}
\newcommand{\red}{$red$}
\newcommand{\purple}{$purple$}
\newcommand{\processed}{processed}

\title{Sparsifying, Shrinking and Splicing for \\Minimum Path Cover in Parameterized Linear Time
\thanks{This work was partially funded by the European Research Council (ERC) under the European Union's Horizon 2020 research and innovation programme (grant agreement No.~851093, SAFEBIO), the Academy of Finland (grants No.~322595, 328877), the US Fulbright program, the Fulbright Finland Foundation, the Helsinki Institute for Information Technology (HIIT), as well as the US National Science Foundation (grant DBI-1759522).}
}

\author[1]{Manuel Cáceres}

\author[1]{Massimo Cairo}

\author[2]{Brendan Mumey}

\author[3]{Romeo Rizzi}

\author[1]{Alexandru~I.~Tomescu}

\affil[1]{\small Department of Computer Science, University of Helsinki, Finland, \texttt{\{manuel.caceresreyes,{alexandru.tomescu\}@helsinki.fi}}}
\affil[2]{\small School of Computer Science, Montana State University, USA, \texttt{brendan.mumey@montana.edu}}
\affil[3]{\small Department of Computer Science, University of Verona, Italy, \texttt{romeo.rizzi@univr.it}}

\date{}

\begin{document}

\maketitle

\begin{abstract}
A \emph{minimum path cover} (MPC) of a directed acyclic graph (DAG) $G = (V,E)$ is a minimum-size set of paths that together cover all the vertices of the DAG. Computing an MPC is a basic polynomial problem, dating back to Dilworth's and Fulkerson's results in the 1950s. Since the size $k$ of an MPC (also known as the \emph{width}) can be small in practical applications, research has also studied algorithms whose complexity is parameterized on $k$.

 We obtain two new MPC parameterized algorithms for DAGs running in time $O(k^2|V|\log{|V|} + |E|)$ and $O(k^3|V| + |E|)$. We also obtain a parallel algorithm running in $O(k^2|V| + |E|)$ parallel steps and using $O(\log{|V|})$ processors (in the PRAM model). Our latter two algorithms are the first solving the problem in parameterized linear time. Finally, we present an algorithm running in time $O(k^2|V|)$ for transforming any MPC to another MPC using less than $2|V|$ distinct edges, which we prove to be asymptotically tight. As such, we also obtain edge sparsification algorithms preserving the width of the DAG with the same running time as our MPC algorithms. 
 
 At the core of all our algorithms we interleave the usage of three techniques: \emph{transitive sparsification}, \emph{shrinking} of a path cover, and the \emph{splicing} of a set of paths along a given path.

\end{abstract}
\thispagestyle{empty}

\newpage
\clearpage
\setcounter{page}{1}

\section{Introduction}

A \emph{Minimum Path Cover (MPC)} of a (directed) graph $G = (V, E)$ is a minimum-sized set of paths such that every vertex appears in some path in the set. While computing an MPC is NP-hard in general, it is a classic result, dating back to Dilworth~\cite{dilworth2009decomposition} and Fulkerson~\cite{fulkerson1956note}, that this can be done in polynomial time on directed acyclic graphs (\emph{DAGs}). Computing an MPC of a DAG has applications in various fields. In bioinformatics, it allows efficient solutions to the problems of multi-assembly~\cite{eriksson2008viral,trapnell2010transcript,rizzi2014complexity,chang2015bridger,liu2017strawberry}, perfect phylogeny haplotyping~\cite{bonizzoni2007linear,gramm2007haplotyping}, and alignment to pan-genomes~\cite{makinen2019sparse,ma2021co}. Other examples include scheduling~\cite{colbourn1985minimizing,desrosiers1995time,bunte2009overview,van2016precedence,zhan2016graph,marchal2018parallel}, computational logic~\cite{bova2015model,gajarsky2015fo}, distributed computing~\cite{tomlinson1997monitoring,ikiz2006efficient}, databases~\cite{Jagadish90}, evolutionary computation~\cite{jaskowski2011formal}, program testing~\cite{ntafos1979path}, cryptography~\cite{mackinnon1985optimal}, and programming languages~\cite{kowaluk2008path}. Since in many of these applications the size $k$ (number of paths, also known as \emph{width}) of an MPC is bounded, research has 
also focused in solutions whose complexity is parameterized by $k$. This approach is also related to the line of research ``FPT inside P''~\cite{giannopoulou2017polynomial} of finding natural parameterizations for problems already in P (see also e.g.~\cite{fomin2018fully,koana2021data,abboud2016approximation}).

MPC algorithms can be divided into those based on a reduction to maximum matching~\cite{fulkerson1956note}, and those based on a reduction to minimum flow~\cite{ntafos1979path}. The former compute an MPC of a \emph{transitive} DAG by finding a maximum matching in a bipartite graph with $2|V|$ vertices and $|E|$ edges. Thus, one can compute an MPC of a transitive DAG in time $O(\sqrt{|V|} |E|)$ with the Hopcroft-Karp algorithm \cite{hopcroft1973n}. Further developments of this idea include the $O(k|V|^2)$-time algorithm of Felsner et~al.~\cite{felsner2003recognition}, and the $O(|V|^2 + k\sqrt{k}|V|)$ and $O(\sqrt{|V|}|E| + k\sqrt{k}|V|)$-time algorithms of Chen and Chen~\cite{chen2008efficient,chen2014graph}.

The reduction to minimum flow consists in building a flow network $\flowG$ from $G$, where a global source $s$ and global sink $t$ are added, and  each vertex $v$ of $G$ is split into an edge $(v^{in}, v^{out})$ of $\flowG$ with a demand (lower bound) of one unit of flow (see \Cref{sec:preliminaries} for details). A minimum-valued (integral) flow of $\flowG$ corresponds to an MPC of $G$, which can be obtained by decomposing the flow into paths. This reduction (or similar) has been used several times in the literature to compute an MPC (or similar object)~\cite{ntafos1979path,mohring1985algorithmic,gavril1987algorithms,Jagadish90,ciurea2004sequential,rademaker2012optimal,pijls2013another,marchal2018parallel}, and it is used in the recent $\bestparameterizedsolution$-time solution of M\"akinen et~al.~\cite{makinen2019sparse}. Furthermore, by noting that a path cover of size $|V|$ is always valid (one path per vertex) the problem can be reduced to maximum flow with capacities at most $|V|$ (see for example \cite[Theorem 4.9.1]{bang2008digraphs}) and it can be solved by using maximum flow algorithms outputting integral solutions. As an example, using the Goldberg-Rao algorithm~\cite{goldberg1998beyond} the problem can be solved in time $\widetilde{O}(|E|\min(|E|^{1/2}, |V|^{2/3})+k|E|)$ (the $k|E|$ term is needed for decomposing the flow into an MPC). More recent maximum flow algorithms~\cite{lee2014path,madry2016computing,liu2020faster,kathuria2020unit,van2021minimum,gao2021fully} provide an abundant options of trade-offs, though none of them leads to a parameterized linear-time solution for the MPC problem. Next, we describe our techniques and results.


\paragraph{Sparsification, shrinking and splicing.}
Across our solutions we interleave three techniques.

\emph{Transitive sparsification} consists in the removal of some transitive edges\footnote{Transitive edges are edges whose removal does not disconnect its endpoints.} while preserving the reachability among vertices, and thus the width of the DAG\footnote{Every edge in an MPC removed by a transitive sparsification can be \emph{re-routed} through an alternative path.}. We sparsify the edges to $O(k|V|)$ only, in overall $O(|E|)$ time, obtaining thus a linear dependency on $|E|$ in our running times. Our idea is inspired by the work of Jagadish~\cite{Jagadish90}, which proposed a compressed index for answering reachability queries in constant time: for each vertex $v$ and path $P$ of an MPC, it stores the last vertex in $P$ that reaches $v$ (thus using overall $k|V|$ space). However, three issues arise when trying to apply this idea \emph{inside} an MPC algorithm: (i) it is dependent on an initial MPC (whereas we are trying to compute one), (ii) it can be computed in only $O(k|E|)$ time~\cite{makinen2019sparse}, and (iii) edges in the index are not necessarily in the DAG. We address (i) by using a suboptimal (but yet bounded) path cover whose gradual computation is interleaved with transitive sparsifications, and we address (ii) and (iii) by keeping only $O(k)$ incoming edges per vertex in a \emph{single linear pass} over the edges. 

By \emph{shrinking} we refer to the process of transforming an arbitrary path cover into an MPC. For example, using the flow network $\flowG$ built from the given path cover, we can search for \emph{decrementing} paths, until obtaining a minimum flow corresponding to an MPC. Given an $O(\log{|V|})$ approximation of an MPC, both algorithms of~\cite{felsner2003recognition,makinen2019sparse} shrink this path cover in a separate step. In both of our algorithms, we do not use shrinking as a separate black-box, but instead interleave shrinking steps in the gradual computation of the MPC. Moreover, in the second algorithm we further guide the search for decrementing paths to amortize the search time to parameterized linear time.

Finally, by \emph{splicing} we refer to the general process of reconnecting paths in a path cover so that (after splicing) at least one of them contains a certain path $D$ as a subpath, while working in time proportional to $|D|$. In particular, we show how to perform splicing to apply the changes required by a decrementing path on a flow decomposition for obtaining an MPC (see \Cref{sec:splicing-algorithm}), and also to reconnect paths for reducing the number of edges used by an MPC (see \Cref{sec:edge-thinning}).

\paragraph{A simple divide-and-conquer approach.}

As a first simple example of sparsification and shrinking interleaved inside an MPC algorithm, in \Cref{sec:DandC} we show how these two techniques enable the first divide-and-conquer MPC algorithm.

\begin{restatable}{theorem}{sparsealgorithm}
\label{thm:main-d-c}
Given a DAG $G = (V,E)$ of width $k$, we compute an MPC in time ${O(k^2|V|\log{|V|}+|E|)}$.
\end{restatable}

\Cref{thm:main-d-c} works by splitting a topological ordering of the vertices in half, and recursing in each half. When combining the MPCs from the two halves, we need to (i)~account for the new edges between the two parts (here we exploit sparsification), and (ii)~efficiently combine the two partial path covers into one for the entire graph (and here we use shrinking). Since this divides the problem in disjoint subgraphs, we also obtain the first linear-time parameterized parallel algorithm.

\begin{restatable}{theorem}{parallelsparsealgorithm}
\label{thm:parallel-d-c}
Given a DAG $G = (V,E)$ of width $k$, we compute an MPC in $O(k^2|V|+|E|)$ parallel steps using $O(\log{|V|})$ single processors in the PRAM model~\cite{wyllie1979complexity}.
\end{restatable}

\paragraph{The first linear-time parameterized algorithm.}
Our second algorithm works on top of the minimum flow reduction, but instead of running a minimum flow algorithm and then extracting the corresponding paths (as previous approaches do~\cite{ntafos1979path,mohring1985algorithmic,Jagadish90,ciurea2004sequential,rademaker2012optimal,pijls2013another,marchal2018parallel,makinen2019sparse}), it processes the vertices in topological order, and incrementally maintains an MPC (i.e.~a flow decomposition) $\pathcover$ of the corresponding induced subgraph. When a new vertex $v$ is processed, $\pathcover$ is used to sparsify the edges incoming to $v$ to at most $k$ (see \Cref{sec:progressive-flows}). After that, the path cover $\pathcover \cup \{(v)\}$ is shrunk by searching for a single decrementing path in the corresponding residual graph. The search is guided by assigning an integer level to each vertex. We amortize the time of performing all the searches to $O(k^3)$ time per vertex, thus obtaining the final $O(k^3|V|+|E|)$ running time. 

\begin{restatable}{theorem}{mpcflow}
\label{thm:mpc-flow}
Given a DAG $G = (V,E)$ of width $k$, we compute an MPC in time ${O(k^3|V|+|E|)}$.
\end{restatable}

The amortization is achieved by guiding the search through the assignment of integer levels to the vertices, which allows to perform the traversal in a \emph{layered} manner, from the vertices of largest level to vertices of smallest level (see \Cref{sec:layered-traversal}). If a decrementing path $D$ is found, $\pathcover \cup \{(v)\}$ is updated by splicing it along $D$ (see \Cref{sec:splicing-algorithm}).

An \emph{antichain} is a set of pairwise non-reachable vertices, and it is a well-known result, due to Dilworth~\cite{dilworth2009decomposition}, that the maximum size of an antichain equals the size of an MPC. Our level assignment defines a series of size-decreasing one-way cuts (\Cref{lemma:optimality-by-invariant-a}). Moreover, by noting that these cuts in the network correspond to antichains (see e.g. \cite{pijls2013another}), the levels implicitly maintain a structure of antichains that \emph{sweep} the graph during the algorithm. The high-level idea of maintaining a collection of antichains has been used previously by Felsner et~al.~\cite{felsner2003recognition} and C\'aceres et~al.~\cite{caceres2021a} for the related problem of computing a maximum antichain. However, apart from being restricted to this related problem, these two approaches have intrinsic limitations. More precisely, Felsner et~al.~\cite{felsner2003recognition} maintain a \emph{tower of right-most antichains} for \emph{transitive} DAGs and $k \leq 4$, mentioning that ``the case $k=5$ already seems to require an unpleasantly involved case analysis''~\cite[p.~359]{felsner2003recognition}. C\'aceres et~al.~\cite{caceres2021a} overcome this by maintaining $O(2^k)$ many \emph{frontier antichains}, and obtaining a linear-time parameterized $O(k^24^k|V| + k2^k|E|)$-time maximum antichain algorithm.

Based on the relation between maximum one-way cuts in the minimum flow reduction and maximum antichains in the original DAG (see for example~\cite{mohring1985algorithmic,pijls2013another,marchal2018parallel}), we obtain algorithms computing a maximum antichain from any of our existing algorithms, preserving their running times (see \Cref{thm:main-antichain}). In particular, by using our second algorithm we obtain an exponential improvement on the function of $k$ of the algorithm of C\'aceres et~al.~\cite{caceres2021a}.


\paragraph{Edge sparsification in parameterized linear time.} Our last result in \Cref{sec:edge-thinning} is a structural result concerning the problem of \emph{edge sparsification} preserving the width of the DAG. Edge sparsification is a general concept that consists in finding spanning subgraphs (usually with significantly less edges) while (approximately) preserving certain property of the graph. For example, \emph{spanners} are distance preserving (up to multiplicative factors) sparsifiers, and it is a well-known result that $(1+\epsilon)$ cut sparsifiers can be computed efficiently~\cite{benczur1996approximating}. We show that if the property we want to maintain is the (exact) width of a DAG, then its edges can be sparsified to less than~$2|V|$. Moreover, we show that such sparsification is asymptotically tight (\Cref{remark:2-tight}), and it can be computed in $O(k^2|V|)$ time if an MPC is given as additional input. Therefore, by using our second algorithm we obtain the following result.

\begin{corollary}
\label{cor:main-sparsification}
Given a DAG $G = (V,E)$ of width $k$, we compute a spanning subgraph $G' = (V, E')$ of $G$ with $|E'| < 2|V|$ and width $k$ in time ${O(k^3|V|+|E|)}$.
\end{corollary}

The main ingredient to obtain this result is an algorithm for transforming any path cover into one of the same size using less than~$2|V|$ distinct edges, a surprising structural result.

\begin{restatable}{theorem}{edgethinning}
\label{thm:edge-thinning}
Let $G = (V,E)$ be a DAG, and let $\pathcover, |\pathcover|=t$ be a path cover of $G$. Then, we compute, in $O(t^2|V|)$ time, a path cover $\pathcover',|\pathcover'|=t$, whose number of \emph{distinct} edges is less than $2|V|$.
\end{restatable}

We obtain \Cref{cor:main-sparsification} by using \Cref{thm:edge-thinning} with an MPC and defining $E'$ as the edges in $\pathcover'$. Our approach adapts the techniques used by Schrijver~\cite{schrijver1998bipartite} for finding a perfect matching in a regular bipartite graph. In our algorithm, we repeatedly search for undirected cycles~$C$ of edges joining vertices of high degree (in the graph induced by the path cover), and splice paths along $C$ (according to the \emph{multiplicty} of the edges of $C$) to remove edges from the path cover.

\paragraph{Paper structure.} \Cref{sec:preliminaries} presents basic concepts, the main preliminary results needed to understand the technical content of this paper, and results related to the three common techniques used in latter sections\footnote{We include the full version of this section in \Cref{sec:extended-preliminaries} for completeness.}. \Cref{sec:DandC,sec:progressive-flows} present our $O(k^2|V|\log{|V|}+|E|)$ and $O(k^3|V|+|E|)$ time algorithms for MPC, respectively\footnote{In \Cref{sec:antichain-structure} we show that our second algorithm implicitly maintains a structure of antichains.}. \Cref{sec:edge-thinning} presents the algorithm of \Cref{thm:edge-thinning}. Omitted proofs can be found in the Appendices.

\section{Preliminaries} \label{sec:preliminaries}

\paragraph{Basics.} We denote by $N^+(v)$ ($N^-(v)$) the set of out-neighbors (in-neighbors) of $v$, and by $I^+(v)$ ($I^-(v)$) the edges outgoing (incoming) from (to) $v$. A graph $S = (V_S, E_S)$ is said to be a \emph{subgraph} of $G$ if $V_S \subseteq V$ and $E_S \subseteq E$. If $V_S = V$ it is called \emph{spanning} subgraph. If $V' \subseteq V$, then $G[V']$ is the subgraph of $G$ \emph{induced} by $V'$, defined as $G[V'] = (V', E_{V'})$, where $E_{V'} = \{(u, v) \in E ~:~ u,v \in V'\}$. A \emph{directed acyclic graph} (DAG) is a directed graph without proper cycles. A \emph{topological ordering} of a DAG is a total order of $V$, $v_1,\ldots, v_{|V|}$, such that for all $(v_{i}, v_{j}) \in E$, $i < j$. A topological ordering can be computed in $O(|V|+|E|)$ time~\cite{kahn1962topological,tarjan1976edge}. If there exists a path from $u$ to $v$, then it is said that $u$ \emph{reaches} $v$. The \emph{multiplicity} of an edge $e\in E$ with respect to a set of paths $\pathcover$, $\mu_{\pathcover}(e)$ (only $\mu(e)$ if $\pathcover$ is clear from the context), is defined as the number of paths in $\pathcover$ that contain $e$, $\mu_{\pathcover}(e) = |\{P\in\pathcover\mid e \in P\}|$. The \emph{width} of a graph $G$, $\width(G)$, is the size of an MPC of $G$. We will work with subgraphs induced by a consecutive subsequence of vertices in a topological ordering. The following lemma shows that we can bound the width of these subgraphs by $k = \width(G)$.

\begin{restatable}[\cite{caceres2021a}]{lemma}{topologicalDoNotIncreaseWidth}
    \label{topologicalDoNotIncreaseWidth}
    Let $G = (V, E)$ be a DAG, and $v_1, \ldots , v_{|V|}$ a topological ordering of its vertices. Then, for all $i, j \in [1\ldots |V|], i \le j$, $\width(G_{i,j}) \le \width(G)$, with $G_{i,j} := G[\{v_{i}, \ldots , v_{j}\}]$.
\end{restatable}

\paragraph{Minimum Flow.}
Given a (directed) graph $G = (V, E)$, a source $s \in V$, a sink $t \in V$, and a function of \emph{lower bounds} or \emph{demands} on its edges $d: E \rightarrow \Nzero$, an $st$-\emph{flow} (or just \emph{flow} when $s$ and $t$ are clear from the context) is a function on the edges $f: E \rightarrow \Nzero$, satisfying $f(e) \ge d(e)$ for all $e \in E$ ($f$ \emph{satisfies the demands}) and $\sum_{e \in I^-(v)} f(e) = \sum_{e \in I^+(v)} f(e)$ for all $v \in V \setminus \{s,t\}$ (\emph{flow conservation}). If a flow exists, the tuple $(G, s, t, d)$ is said to be a \emph{flow network}. The \emph{size} of $f$ is the net amount of flow exiting $s$, formally $|f| = \sum_{e \in I^+(s)} f(e) - \sum_{e \in I^-(s)} f(e)$. An $st$-\emph{cut} (or just \emph{cut} when $s$ and $t$ are clear from the context) is a partition $(S, T)$ of $V$ such that $s \in S$ and $t \in T$. An edge $(u, v)$ \emph{crosses} the cut $(S, T)$ if $u\in S$ and $v \in T$, or vice versa. If there are no edges \emph{crossing} the cut from $T$ to $S$, that is, if $\{(u,v) \in E \mid u \in T, v \in S\} = \emptyset$, then $(S, T)$ is a \emph{one-way cut} (ow-cut). The \emph{demand} of an ow-cut is the sum of the demands of the edges crossing the cut, formally $d((S,T)) = \sum_{e = (u, v), u \in S, v \in T} d(e)$. An ow-cut whose demand is maximum among the demands of all ow-cuts is a \emph{maximum ow-cut}.

Given a \emph{flow network} $(G, s, t, d)$, the problem of \emph{minimum flow} consists of finding a flow $f^*$ of minimum size $|f^*|$ among the flows of the network, such flow is a \emph{minimum flow}. If a minimum flow exists, then $(G, s, t, d)$ is a \emph{feasible} flow network. It is a known result~\cite{ahujia1993network,ciurea2004sequential,bang2008digraphs} that the demand of a maximum ow-cut equals the size of a minimum flow.

Given a flow $f$ in a feasible flow network $(G, s, t, d)$, the \emph{residual network} of $G$ with respect to $f$ is defined as $\residual{G}{f} = (V, E_f)$ with $E_f = \{(u, v) \mid (v, u) \in E\} \cup \{e \mid f(e) > d(e)\}$, that is, the \emph{reverse edges} of $G$, plus the edges of $G$ on which the flow can be decreased without violating the demands (\emph{direct edges}). Note that a path from $s$ to $t$ in $\residual{G}{f}$ can be used create another flow $f'$ of smaller size by increasing flow on reverse edges and decreasing flow on direct edges of the path, such path it is called \emph{decrementing path}. A flow $f$ is a minimum flow if and only if there is no decrementing path in $\residual{G}{f}$ (see \Cref{sec:min-flow-extended}). A \emph{flow decomposition} of $f$ is a set of $|f|$ paths $\pathcover$ in $G$ such that $f(e) = \mu_{\pathcover}(e)$ for all $e \in E$, in this case it is said that $f$ is the flow \emph{induced} by $\pathcover$. If $\pathcover$ is a flow decomposition of $f$, then the residual network of $G$ with respect to $\pathcover$ is $\residual{G}{f}$.

\paragraph{MPC in DAGs through Minimum Flow.}
The problem of finding an MPC in a DAG $G = (V, E)$  can be solved by a reduction to the problem of minimum flow on an appropriate feasible flow network $(\flowG = (\flowV, \flowE), s, t, d)$~\cite{ntafos1979path}, defined as: $\flowV = \{s, t\} \cup \{v^{in} \mid v \in V\} \cup \{v^{out} \mid v \in V\}$ ($\{s,t\}\cap V = \emptyset$); $\flowE = \{(s, v^{in}) \mid v \in V\} \cup \{(v^{out}, t) \mid v \in V\} \cup \{(v^{in}, v^{out}) \mid v \in V\} \cup \{(u^{out}, v^{in}) \mid (u,v) \in E\}$; and $d(e) = 1$ if $e = (v^{in}, v^{out})$ for some $v \in V$ and $0$ otherwise. The tuple $(\flowG, s, t, d)$ is the \emph{flow reduction} of $G$. Note that $|\flowV| = O(|V|)$, $|\flowE|=O(|E|)$, and $\flowG$ is a DAG. Every flow $f$ of $(\flowG, s, t, d)$ can be decomposed into $|f|$ paths corresponding to a path cover of $G$ (by removing $s$ and $t$ and merging the edges $(v^{in}, v^{out})$ into $v$, see \Cref{sec:minflow-reduction-extended}). A minimum flow of $(\flowG, s, t, d)$  has size $\width(G)$, thus providing an MPC of $G$ after decomposing it (see \Cref{sec:minflow-reduction-extended}). Moreover, the set of edges of the form $(v^{in}, v^{out})$ crossing a maximum ow-cut corresponds to a maximum antichain of $G$ (by merging the edges $(v^{in}, v^{out})$ into $v$, see~\cite{mohring1985algorithmic,rademaker2012optimal,pijls2013another,marchal2018parallel}). By further noting that if $f$ is a minimum flow of $(\flowG, s, t, d)$, and defining $S = \{v \in \flowV \mid s\text{ reaches } v\text{ in }\residual{\flowG}{f}\}$, then $(S, T=\flowV\setminus S)$ corresponds to a maximum ow-cut, we obtain the following result.
\begin{restatable}{lemma}{fastantichain}
\label{thm:main-antichain}
Given a DAG $G = (V,E)$ of width $k$ and an MPC $\pathcover$, we compute a maximum antichain of $G$ in time $O(k|V|+|E|)$.
\end{restatable}

As such, this allows us to obtain algorithms computing a maximum antichain from any of our MPC algorithms, preserving their running times.

\paragraph{Sparsification, shrinking, splicing.} We say that a spanning subgraph $S = (V, E_S)$ of a DAG $G = (V, E)$ is a \emph{transitive sparsification} of $G$, if for every $u, v \in V$, $u$ reaches $v$ in $S$ if and only if $u$ reaches $v$ in $G$. Since $G$ and $S$ have the same reachability relations on their vertices, they share their antichains, thus $\width(G) = \width(S)$. As such, an MPC of $S$ is also an MPC of $G$, thus the edges $E\setminus E_S$ can be safely removed for the purpose of computing an MPC of $G$. If we have a path cover $\pathcover$ of size $t$ of $G$, then we can \emph{sparsify} (remove some transitive edges) the incoming edges of a particular vertex $v$ to at most $t$ in time $O(t+|N^-(v)|)$. If $v$ has more than $t$ in-neighbors then two of them belong to the same path, and we can remove the edge from the in-neighbor appearing first in the path. We create an array of $t$ elements initialized as $survivor \gets (v_{-\infty})^t$, where $v_{-\infty} \not \in V$ is before every $v\in V$ in topological order. Then, we process the edges $(u,v)$ incoming to $v$, we set $i \gets path(u)$ ($path(u)$ gives the ID of some path of $\pathcover$ containing $u$) and if $survivor[i]$ is before $u$ in topological order we replace it $survivor[i] \gets u$. Finally, the edges in the sparsification are $\{(survivor[i], v) \mid i\in\{1,\ldots,t\} \land survivor[i] \not = v_{-\infty}\}$.

\begin{restatable}{obs}{obssparsification}
\label{obs:sparsificationVertex}
Let $G = (V, E)$ be a DAG, $\pathcover$ a path cover, $|\pathcover|=t$, $v$ a vertex of $G$, and $path:V\rightarrow \{1,\ldots,t\}$ a function  that answers in constant time $path(v)$, the ID of some path of $\pathcover$ containing $v$. We can sparsify the incoming edges of $v$ to at most $t$ in time $O(t+|N^-(v)|)$.
\end{restatable}

By first computing a $path$ function, and then applying \Cref{obs:sparsificationVertex} to every vertex we obtain.

\begin{restatable}{lemma}{sparsificationalgorithm}
    \label{sparsification_algorithm}
    Let $G = (V, E)$ be a DAG, and $\pathcover$, $|\pathcover|=t$, be a path cover of $G$. Then, we can sparsify $G$ to $S = (V, E_S)$, such that $\pathcover$ is a path cover of $S$ and $|E_S| \le t|V|$, in $O(t|V| + |E|)$ time.
\end{restatable}

The following lemma shows that we can locally sparsify a subgraph and apply these changes to the original graph to obtain a transitive sparsification.
\begin{restatable}{lemma}{sparsificationOfSubgraph}
    \label{sparsificationOfSubgraph}
    Let $G = (V, E)$ be a graph, $S = (V_S, E_S)$ a subgraph of $G$, and $S' = (V_S, E_{S'})$ a transitive sparsification of $S$. Then $G' = (V, E\setminus (E_S \setminus E_{S'}))$ is a  transitive sparsification of $G$.
\end{restatable}

As explained before, shrinking is the process of transforming an arbitrary path cover $\pathcover$ into an MPC, and it can be solved by finding $|\pathcover|-\width(G)$ decrementing paths in $\residual{\flowG}{\pathcover}$, and then decomposing the resulting flow into an MPC. M\"akinen et al.~\cite{makinen2019sparse} apply this idea to shrink a path cover of size $O(k\log{|V|})$. We generalize this approach in the following lemma.

\begin{restatable}{lemma}{shrinkingalgorithm}
\label{result:shrinking}
    Given a DAG $G = (V, E)$ of width $k$, and a path cover $\pathcover$, $|\pathcover|=t$, of $G$, we can obtain an MPC of $G$ in time $O(t(|V| + |E|))$.
\end{restatable}

As said before, splicing consists in reconnecting paths in a path cover $\pathcover$ so that (after reconnecting) at least one of the paths contains as a subpath a certain path $D$, in time $O(|D|)$. Splicing additionally requires that for every edge $e$ of $D$ there is at least one path in $\pathcover$ containing $e$.

\begin{restatable}{lemma}{splicingalgorithm}
\label{result:splicing}
    Let $G = (V, E)$ be a DAG, $D$ a proper path, and $\pathcover$ path cover such that for every edge $e \in D$ there exists $P\in\pathcover, e \in P$. We obtain a path cover $\pathcover'$ of $G$ such that $|\pathcover'| = |\pathcover|$ and there exists $P\in\pathcover'$ containing $D$ as a subpath, in time $O(|D|)$. Moreover, $\mu_{\pathcover}(e) = \mu_{\pathcover'}(e)$ for all $e \in E$.
\end{restatable}

Because of the last property of $\pathcover'$, the flow induced by $\pathcover$ is the same as the flow induced by $\pathcover'$. As such, if $\pathcover$ is a flow decomposition of a flow $f$, then $\pathcover'$ is also a flow decomposition of $f$.

\section{Divide and Conquer Algorithm}\label{sec:DandC}

\begin{figure}[t]
    \centering
    \begin{subfigure}[b]{0.24\textwidth}
        \centering
        \includegraphics[width=0.95\textwidth]{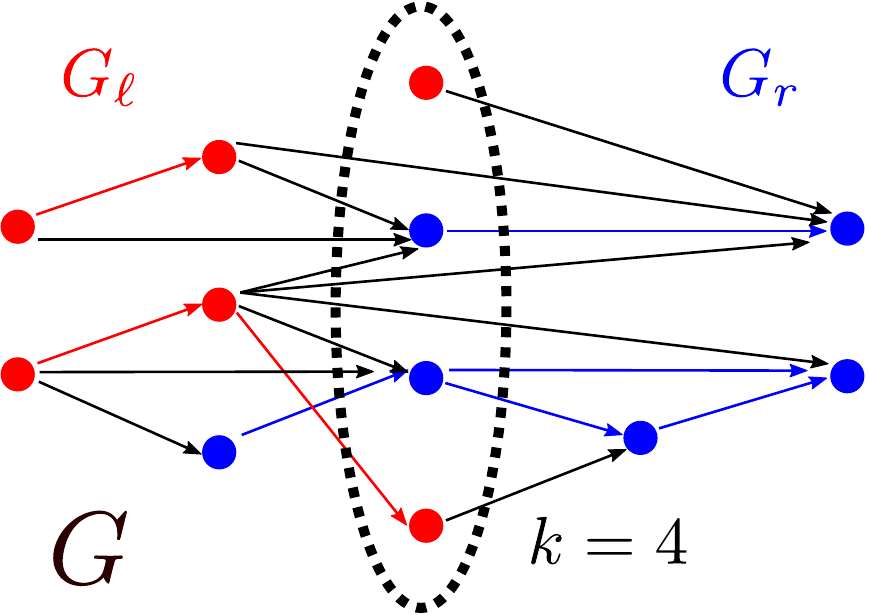}
        \caption[]%
        {{\small Input graph}}
        \label{subfig:input-graph}
    \end{subfigure}
    \begin{subfigure}[b]{0.24\textwidth}
        \centering
        \includegraphics[width=0.95\textwidth]{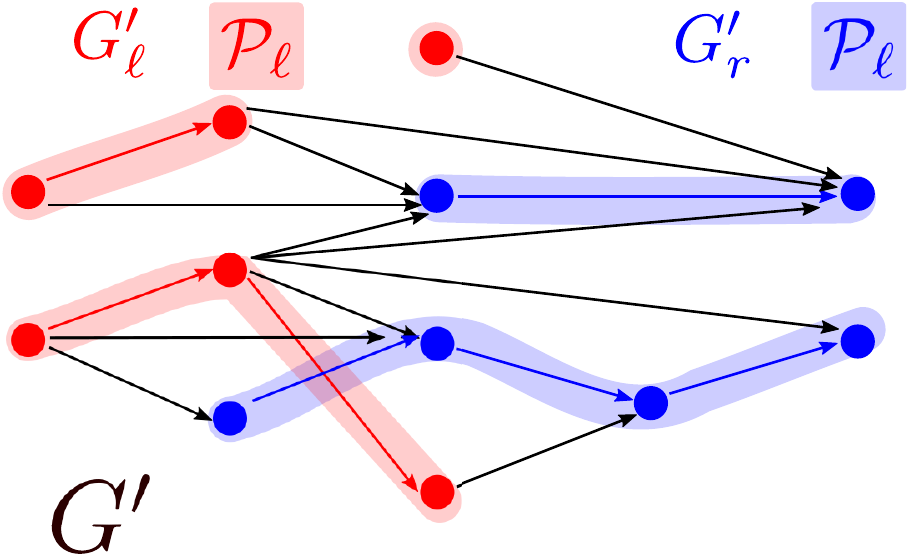}
        \caption[]%
        {{\small Result of recursion}}
        \label{subfig:recursion}
    \end{subfigure}
    \begin{subfigure}[b]{0.24\textwidth}
        \centering
        \includegraphics[width=0.95\textwidth]{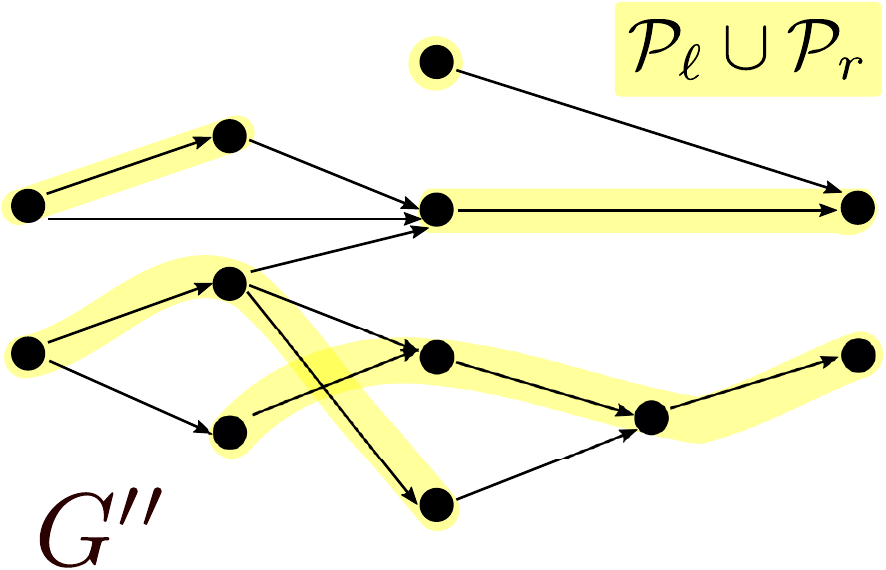}
        \caption[]%
        {{\small Result of sparsification}}
        \label{subfig:sparsification}
    \end{subfigure}
    \begin{subfigure}[b]{0.24\textwidth}
        \centering
        \includegraphics[width=0.95\textwidth]{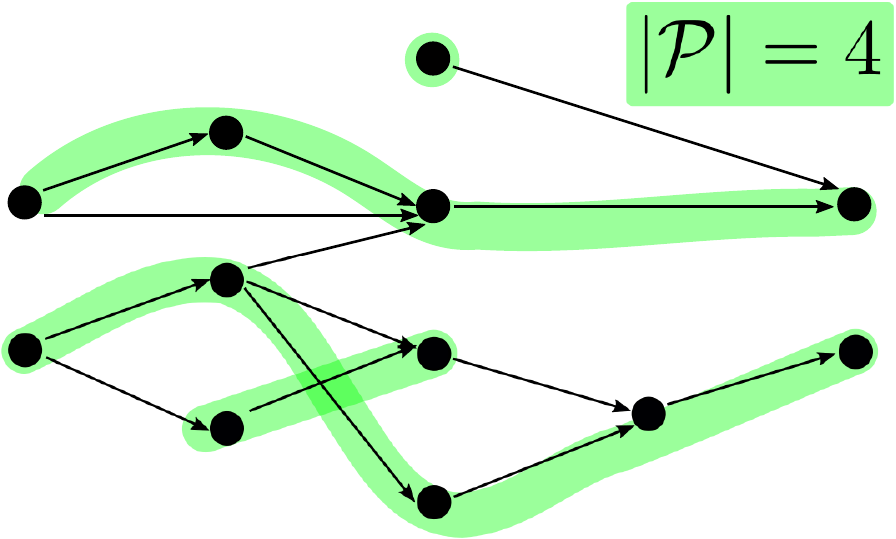}
        \caption[]%
        {{\small Result of shrinking}}
        \label{subfig:shrinking}
    \end{subfigure}
    \caption[]
    {\small Main steps of the divide-and-conquer algorithm applied to a DAG $G$. \Cref{subfig:input-graph} shows the input graph, a maximum antichain, and the division into $G_{\ell}$ and $G_r$. \Cref{subfig:recursion} shows the resulting graph $G'$ after applying the algorithm recursively into $G_{\ell}$ and $G_r$, the corresponding sparsifications $G_{\ell}'$ and $G_r'$, and path covers $\pathcover_{\ell}$ and $\pathcover_r$. \Cref{subfig:sparsification} shows the result $G''$ of the sparsification algorithm run on $G'$ with the paths $\pathcover_\ell \cup \pathcover_r$. \Cref{subfig:shrinking} shows the result $\pathcover$ after shrinking.}
     \label{fig:divide-and-conquer}
\end{figure}

\sparsealgorithm*

\begin{proof}
    Before starting the recursion compute a topological ordering of the vertices $v_1, \ldots , v_{|V|}$ in time $O(|V|+|E|)$. Solve recursively in the subgraph $G_{\ell} = (V_{\ell}, E_{\ell})$ induced by $v_1, \ldots , v_{|V|/2}$, obtaining an MPC $\pathcover_{\ell}$ of a sparsification $G'_{\ell} = (V_{\ell}, E'_{\ell})$ of $G_{\ell}$ with $|E'_{\ell}| \le 2|\pathcover_{\ell}||V_{\ell}|$, and in the subgraph $G_r = (V_r, E_r)$ induced by $v_{|V|/2+1}, \ldots , v_{|V|}$, obtaining an MPC $\pathcover_{r}$ of a sparsification $G'_r = (V_r, E'_r)$ of $G_r$ with $|E'_r| \le 2|\pathcover_{r}||V_r|$. By \Cref{topologicalDoNotIncreaseWidth}, $|\pathcover_\ell|\le k$ and $|\pathcover_r|\le k$. Applying \Cref{sparsificationOfSubgraph} with $G'_{\ell}$ and $G'_{r}$ we obtain that $G' = (V, E' = E'_{\ell} \cup E'_{r} \cup E_{\ell r})$ is a sparsification of $G$ with $|E'| \le 2|\pathcover_{\ell}||V_{\ell}| + 2|\pathcover_r||V_r| + |E_{\ell r}| \le |E_{\ell r}| + 2k|V|$, where $E_{\ell r}$ are the edges in $G$ from $V_{\ell}$ to $V_r$. We consider the path cover $\pathcover_{\ell} \cup \pathcover_r$ of $G'$ and use \Cref{sparsification_algorithm} to obtain a sparsification $G'' = (V, E'')$ of $G'$ in time $O(|E'| + (|\pathcover_{\ell}|+|\pathcover_r|)|V|) = O(|E_{\ell r}| + k|V|)$ such that $|E''| \le (|\pathcover_{\ell}|+|\pathcover_r|)|V| \le 2k|V|$. Finally, we \emph{shrink} $\pathcover_{\ell} \cup \pathcover_r$ in $G''$ to $\pathcover$ of size $k$ in $O((|\pathcover_{\ell}|+|\pathcover_r|)k|V|) = O(k^2|V|)$ time (\Cref{result:shrinking}). 

    The complexity analysis considers the recursion tree of the algorithm. Note that the complexity of a recursion step is $O(k^2|V| + |E_{\ell r}|)$, that is, every vertex of the corresponding subgraph costs $O(k^2)$ and every edge going from the left subgraph to the right subgraph costs $O(1)$. Since the division of the graph generates disjoint subgraphs, every vertex appears in $O(\log{|V|})$ nodes in the recursion tree, and every edge going from left to right appears in exactly one node in the recursion tree. Therefore, the total cost is $O(|E|+k^2|V|\log{|V|})$. \Cref{fig:divide-and-conquer} illustrates the algorithm. 
\end{proof}

 Since our algorithm is based on divide and conquer, we can parallelize the work done on every sub-part of the input, and obtain a linear-time parallel algorithm for the MPC problem.

\parallelsparsealgorithm*

\section{Progressive Flows Algorithm}\label{sec:progressive-flows}

In this section we prove \Cref{thm:mpc-flow}. To achieve this result we rely on the reduction from MPC in a DAG to minimum flow (see \Cref{sec:preliminaries}). We process the vertices of $G$ one by one in a topological ordering $v_1, \ldots, v_{|V|}$. At each step, we maintain a set of \emph{st-flow paths} $\pathcover_i$ that corresponds to a flow decomposition of a minimum flow of $\flowG_i = (\flowV_i, \flowE_i)$ (the flow reduction of $G_i = G[\{v_1, \ldots, v_i\}]$), that is, an MPC of $G_i$. When the next vertex $v_{i+1}$ is considered, we first use $\pathcover_{i}$ to sparsify its incoming edges to at most $|\pathcover_{i}| = O(k)$ in time $O(k+|N^-(v_{i+1})|)$ (see \Cref{obs:sparsificationVertex} and \Cref{topologicalDoNotIncreaseWidth}). Then, we set $\tmp_{i+1} \gets \pathcover_i \cup \{(v_{i+1}^{in}, v_{i+1}^{out})\}$, where $(v_{i+1}^{in}, v_{i+1}^{out})$ corresponds to the edge representing $v_{i+1}$ in the flow reduction (we represent $st$-flow paths either as a sequence of vertices or edges excluding the extremes for convenience). $\tmp_{i+1}$ represents a path cover of $G_{i+1}$, and we use it to try to find a decrementing path in $\residual{\flowG_{i+1}}{\tmp_{i+1}}$. If such decrementing path $D$ is found, some flow paths along $D$ are \emph{spliced} to generate $\pathcover_{i+1}$, such that $|\pathcover_{i+1}| = |\pathcover_{i}|$ (see \Cref{sec:splicing-algorithm}). Otherwise, if no decrementing path is found, we set $\pathcover_{i+1} \gets \tmp_{i+1}$.

We guide the traversal for a decrementing path by assigning an integer level $\ell(v)$ to each vertex $v$ in $\flowG_{i}$. The search is performed in a \emph{layered} manner: it starts from the highest reachable layer (the vertices of highest level according to $\ell$), and it only continues to the next highest reachable layer once all reachable vertices from the current layer have been visited (see \Cref{sec:layered-traversal}). To allow the \emph{layered traversal} and to achieve amortized $O(k^3)$ time per vertex, we maintain three invariants in the algorithm (see \Cref{sec:levels-and-invariants}) and update the level assignment accordingly (see \Cref{sec:level-updates}).

\subsection{Levels, layers and invariants}\label{sec:levels-and-invariants}
We define the level assignment given to the vertices of $\flowG_{i}$, $\ell: \flowV_i \to \{0, 1, \ldots, \width(G_i)\} \cup \{-\infty, +\infty\}$, and the invariants maintained on $\ell$. A \emph{layer} is a maximal set of vertices with the same level, thus layer $l$ is $\{v \in V(\flowG_{i}) \mid \ell(v) = l\}$. All layers form a partition of $\flowV_i$. We extend the definition of level assignment to paths, the level of a path is the maximum level of a vertex in the path, that is, if $P$ is a path of $\flowG_i$, then $\ell(P) = \max_{v \in P} \ell(v)$. We define $\pathcover_i^{\ge l} \subseteq \pathcover_i$, as the flow paths whose level is at least $l$, $\pathcover_i^{\ge l} = \{P \in \pathcover_i \mid \ell(P)\ge l\}$. Note that $|\pathcover_i^{\ge l}| \ge |\pathcover_i^{\ge l'}|$ if $l' > l$. 

At the beginning we fix $\ell(s) = -\infty$ and $\ell(t) =  +\infty$. We also maintain that $0\le\ell(v)\le\width(G_i)$ for all $v \in \flowV_i \setminus \{s,t\}$. Additionally, we maintain the following invariants:
\begin{description}
    \item[Invariant A]: If $(u, v)$ is an edge in $\residual{\flowG_i}{\pathcover_i}$ and $\{u, v\} \cap \{s, t\} = \emptyset$, then $\ell(u) \ge \ell(v)$. 
    \item[Invariant B]: If $(u^{in}, u^{out})$ is the last edge of some $P \in \pathcover_i$, then $\ell(u^{in}) < \ell(u^{out})$.
    \item[Invariant C]: If $l, l'$ are positive integers with $l' > l$, then $|\pathcover_i^{\ge l}| > |\pathcover_i^{\ge l'}|$. 
\end{description}

Note that, since we do not include $s$ and $t$ in the representation of flow paths, $0\le\ell(P)\le\width(G_i)$ for all $P \in \pathcover_{i}$, moreover, by \invB{}, $\ell(P) \ge 1$, thus $\pathcover_i^{\ge 1} = \pathcover_i$. Also note that \invC{} implies that every layer $l \in \{1,\ldots,L\}, L =\max_{v\in\flowV_i\setminus\{t\}}\ell(v)$ is not empty.

\subsection{Progressive flows algorithm}\label{sec:progressive-flows-algorithm}
Our algorithm starts by using $\pathcover_i$ to obtain at most $|\pathcover_i|$ edges incoming to $v_{i+1}$ in time $O(|\pathcover_i|+|N^-(v_{i+1})|) = O(k+|N^-(v_{i+1})|)$ (see \Cref{obs:sparsificationVertex}). This procedure requires to answer $path(v)$ (the ID of some path of $\pathcover_{i}$ containing $v$) queries in constant time. To satisfy this requirement, we maintain path IDs on every vertex/edge of every flow path $P \in \pathcover_i$. In each iteration of our algorithm, these path IDs can be broken by the splicing algorithm (\Cref{sec:splicing-algorithm}) but are repaired before the beginning of the next iteration (\Cref{sec:level-updates}). The following lemma states that the sparsification of incoming edges in $G_{i+1}$ produces an sparsification of outgoing edges in the residual.

\begin{lemma}\label{lemma:sparsifiedresidual}
For every $x \in \flowV_{i+1}\setminus \{s, t\}$, $|N^+({x})| = O(|\pathcover_{i}|)$, in $\residual{\flowG_{i+1}}{\tmp_{i+1}}$.
\end{lemma}
\begin{proof}
If $x$ is of the form $v^{in}$, then its only direct edge could be $(v^{in}, v^{out})$ (only if $(v^{in}, v^{out})$ appears in more than one path in $\pathcover_{i+1}$), its reverse edges are of the form $(v^{in}, u^{out})$, such that $(u, v)$ is an edge in $G_{i+1}$, thus there are at most $|\pathcover_i|$ of such edges because of sparsification (recall that $|\pathcover_j| \le |\pathcover_i|$ for $j < i$, by \Cref{topologicalDoNotIncreaseWidth}). On the other hand, if $x$ is of the form $u^{out}$, then the only reverse edge is $(u^{out}, u^{in})$. To bound the number of direct edges consider the $st$-ow-cut $(S, T)$, with $S =\{v \in \flowV_{i+1} : \mbox{$v$ reaches $u^{out}$ in $\flowG_{i+1}$}\}$. The flow induced by $\tmp_{i+1}$ crossing the cut cannot be more that $|\tmp_{i+1}| = |\pathcover_i|+1$, and thus the number of direct edges $(u^{out}, v^{in})$ is at most $|\pathcover_i|+1$.
\end{proof}


\subsubsection{Layered traversal}\label{sec:layered-traversal}
Our layered traversal performs a BFS in each reachable layer from highest to lowest. If $t$ is reached, the search stops and the algorithm proceeds to splice the flow paths along the decrementing path found. Since $\pathcover_{i}$ represents a minimum flow of $\flowG_i$, every decrementing path $D$ in $\residual{\flowG_{i+1}}{\tmp_{i+1}}$ starts with the edge $(s, v_{i+1}^{in})$ and ends with an edge of the form $(u^{out}, t)$ such that some flow path of $\pathcover_i$ ends at $u^{out}$. Moreover, since $(v_{i+1}^{in}, v_{i+1}^{out})$ does not exist in $\residual{\flowG_{i+1}}{\tmp_{i+1}}$, the second edge of $D$ must be a reverse edge of the form $(v_{i+1}^{in}, u^{out})$, such that $u$ is an in-neighbor of $v_{i+1}$ in $G_{i+1}$.

We work with $|\pathcover_{i}|+1$ queues $Q_{0}, Q_1, \ldots, Q_{|\pathcover_{i}|}$ (one per layer), where $Q_j$ contains the \emph{enqueued} elements from layer $j$, therefore it is initialized as $Q_j \gets \{u^{out} \mid (u^{out}, v_{i+1}^{in}) \in \flowE_{i+1} \land \ell(u^{out}) = j\}$. By \Cref{lemma:sparsifiedresidual}, this initialization takes $O(|\pathcover_{i}|) = O(k)$ time, and it is charged to $v$. We start working with $Q_{|\pathcover_i|}$. When working with $Q_j$, we obtain the first element $u$ from the queue (if no such element exists we move to layer $j-1$ and work with $Q_{j-1}$), then we \emph{visit} $u$ and for each non-visited out-neighbor $v$ we add $v$ to $Q_{\ell(v)}$. Adding the out-neighbors of $v$ to the corresponding queues is charged to $v$, which amounts to $O(|\pathcover_i|) = O(k)$ by \Cref{lemma:sparsifiedresidual}. Since edges in the residual do not increase the level (\invA), out-neighbors can only be added to queues at an equal or lower layer. As such, this traversal advances in a \emph{layered} manner, and it finds a decrementing path if one exists.

Note that the running time of the layered traversal can be bounded by $O(|\pathcover_i|)=O(k)$ per visited vertex. If no decrementing path is found we update the level of the vertices as explained in \Cref{sec:level-updates}. Otherwise, we first splice flow paths along the decrementing path $D$ (\Cref{sec:splicing-algorithm}).

\subsubsection{Splicing algorithm}\label{sec:splicing-algorithm}
Given a decrementing path $D$ in $\residual{\flowG_{i+1}}{\tmp_{i+1}}$, we splice flow paths along $D$ to obtain $\pathcover_{i+1}$. Reverse edges in $D$ indicate that we should push $1$ unit of flow in the opposite direction, thus an edge representing this flow unit should be created. On the other hand, direct edges in $D$ indicate that we should subtract $1$ unit of flow from that edge, in other words, that this edge should be removed from some flow path containing it. As explained in \Cref{sec:layered-traversal}, $D$ starts by a direct edge $(s, v_{i+1}^{in})$, followed by a reverse edge $(v_{i+1}^{in}, u^{out})$ such that $(u, v_{i+1})$ is an edge in $G_{i+1}$. It then continues by a (possibly empty) sequence of reverse and direct edges, and it finishes by a direct edge $(u^{out}, t)$, such that some flow path of $\pathcover_{i}$ ends at $u^{out}$.

A \emph{direct (reverse) segment} is a maximal subpath of \emph{direct (reverse)} edges of $D$. The \emph{splicing} algorithm processes direct and reverse segments interleaved as they appear in $D$. It starts by processing the first reverse segment (the one starting with $(v_{i+1}^{in}, u^{out})$). The procedure that process reverse segments receives as input the suffix of a flow path (the first call receives $((v_{i+1}^{in}, v_{i+1}^{out}))$). It creates the corresponding flow subpath (the reverse of the segment), appends it to the path that received as input, and provides the resulting path as input of the procedure handling the next direct segment. The procedure that handles direct segments $S$, also receives as input the suffix of a flow path. It \emph{splices} the paths of the flow decomposition along $S$ using the procedure of \Cref{result:splicing}, obtaining a new flow decomposition such that one of the paths $P$ contains $S$ as a subpath. It then removes $S$ from $P$ and reconnects the prefix of $P$ before $S$ with the path given as input, and provides the suffix of $P$ after $S$ as input of the procedure handling the next reverse segment.

Note that both procedures run in time proportional to the corresponding segment (see \Cref{result:splicing} for direct segments). As such, the splicing algorithm takes $O(D)$ time. Moreover, since all vertices in the decrementing path are also vertices visited by the traversal, the running time is bounded above by the running time of the layered traversal, that is, $O(k)$ per visited vertex.

\begin{figure}[t]
    \centering
    \begin{subfigure}[b]{0.4\textwidth}
        \centering
        \includegraphics[width=0.8\textwidth]{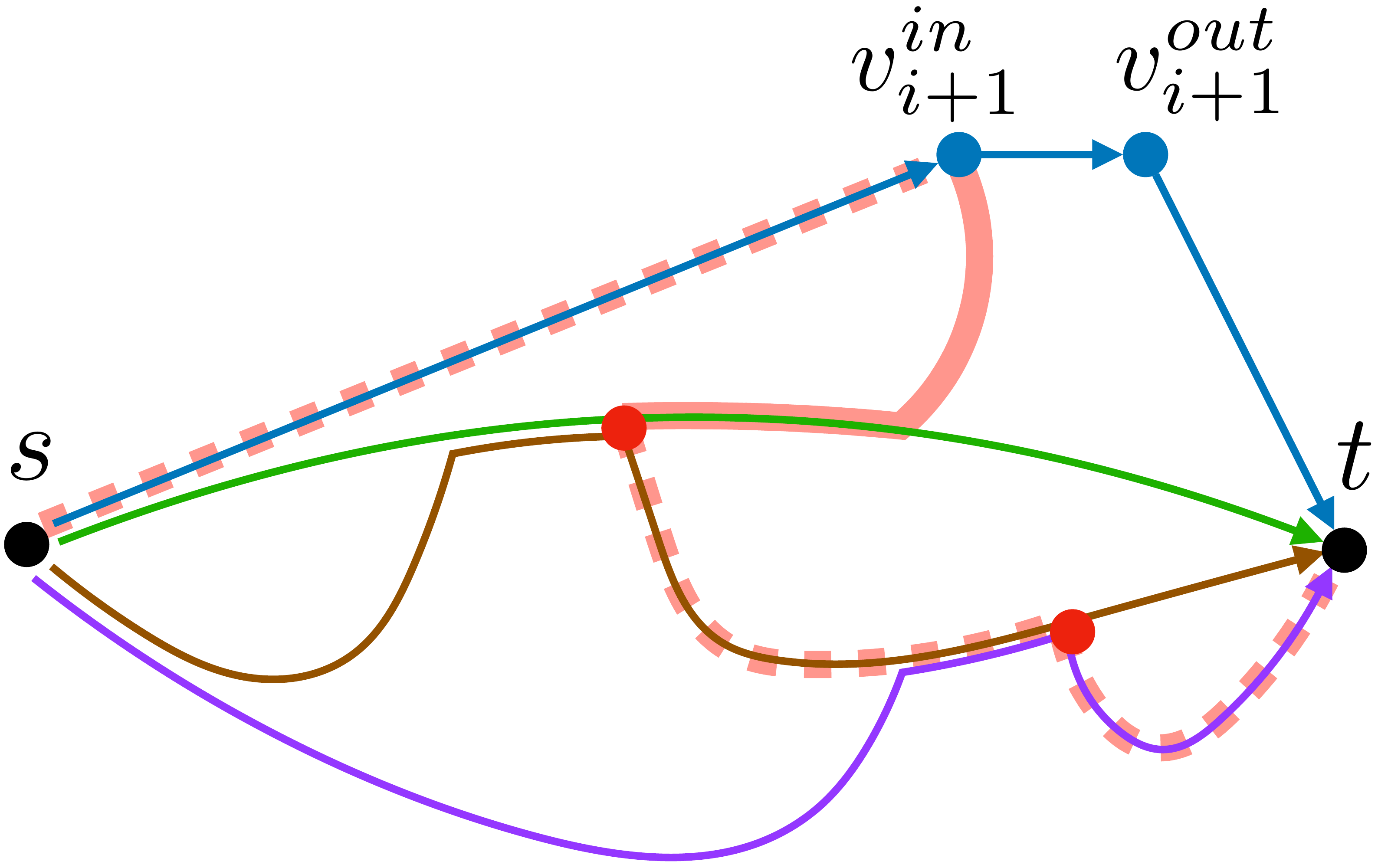}
        \caption[]%
        {{\small Before splicing}}
        \label{subfig:before}
    \end{subfigure}
    \hspace{1cm}
    \begin{subfigure}[b]{0.4\textwidth}
        \centering
        \includegraphics[width=0.8\textwidth]{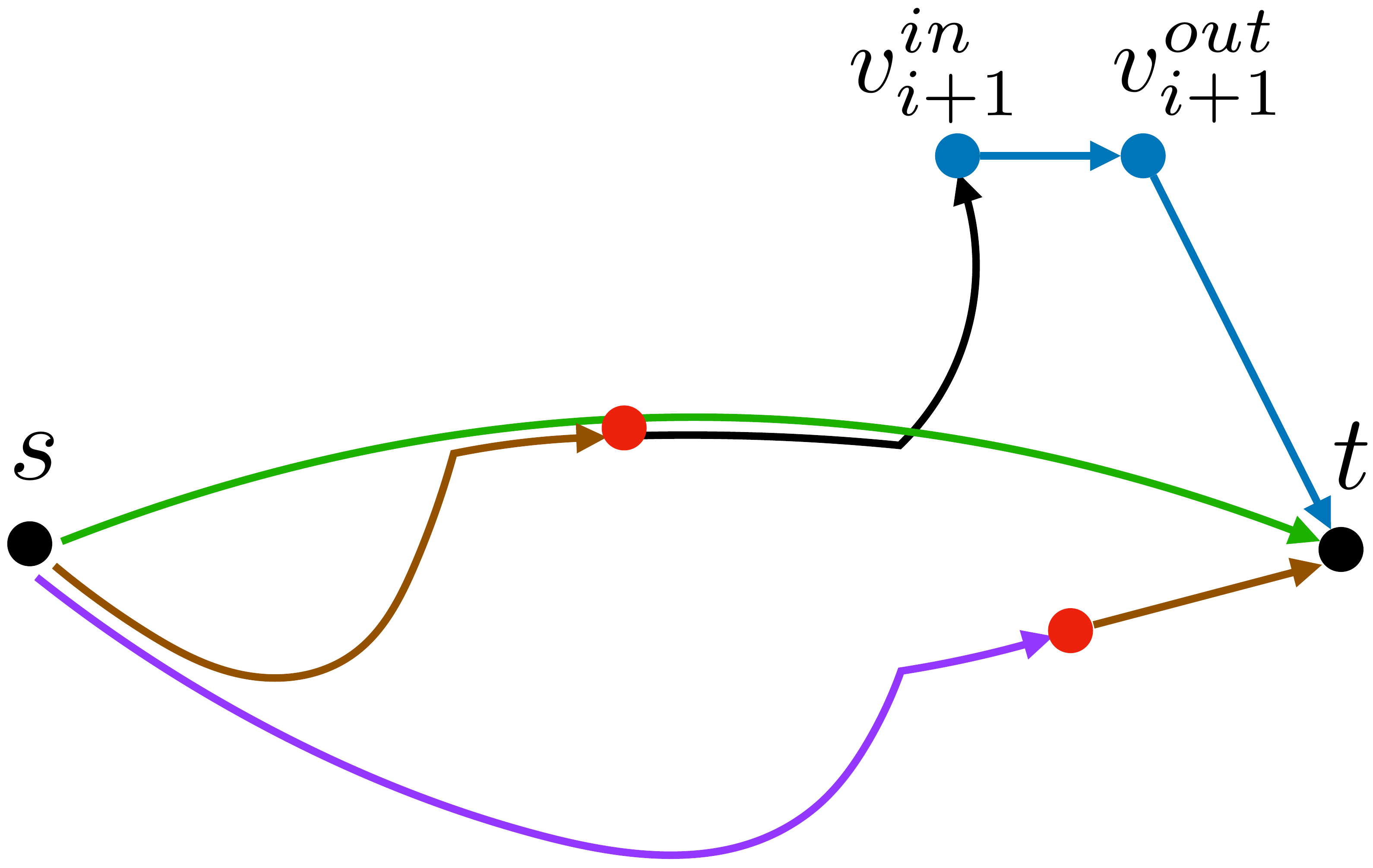}
        \caption[]%
        {{\small After splicing}}
        \label{subfig:after}
    \end{subfigure}
    \caption[]
    {\small Effect of the splicing along a decrementing path $D$ of $\residual{\flowG_{i+1}}{\tmp_{i+1}}$. We only show vertices $s, t, v^{in}_{i+1}, v^{out}_{i+1}$, just four flow paths in blue, green, brown and purple (with some overlap), and two red vertices where splicing of flow paths occurs (splicing points). \Cref{subfig:before} shows the four flow paths before splicing. Path $D$ is highlighted in dashed red (direct segments) and solid red (reverse segment). \Cref{subfig:after} shows that splicing along $D$ transforms the four flow paths into three. The reverse segment creates a subpath (black) of one of these. The direct segments remove subpaths of previous flow paths. The splicing points now join subpaths of the previous brown and blue, and purple and brown, paths respectively.}
    \label{fig:splicing}
\end{figure}

\Cref{fig:splicing} illustrates the effect of the splicing algorithm on flow paths.




 \subsubsection{Level and path updates}\label{sec:level-updates}

After obtaining $\pathcover_{i+1}$, we update the level of some vertices of $\flowV_{i+1}$ to maintain the invariants (\Cref{sec:levels-and-invariants}) of the level assignment $\ell$. Moreover, to sparsify (\Cref{sec:progressive-flows-algorithm}) in the next iteration, we also repair the path IDs on the vertices/edges of $\pathcover_{i+1}$ that could be in an inconsistent state after running the splicing algorithm.

If the smallest layer visited during the traversal is layer $l$, then we set $\ell(v_{i+1}^{in}) = l$, $\ell(v_{i+1}^{out}) = l+1$ (to maintain \invB{}, see \Cref{sec:invariants}), and change the level of every vertex $u$ visited during the traversal to $\ell(u) = l$ (to maintain \invA{}, see \Cref{sec:invariants}).

If a decrementing path was found (and the splicing algorithm was executed) we first repair the path IDs by traversing every flow path of $\pathcover_{i+1}$ backwards from the last vertex, until we arrive to a vertex of level less than $l$, from which we obtain the corresponding path ID that we then update by going back (forwards) in the flow path. After that, the following observations hold.

\begin{obs}\label{obs:end-vertices}
Let $E_{i} = \{u^{out} \in \flowV_i \mid \exists P \in \pathcover_{i}, \text{$u^{out}$ is the last vertex of $P$}\}$, and $A_{i+1}$ the singleton set containing the last vertex in the decrementing path found by the layered traversal in $\residual{\flowG_{i+1}}{\tmp_{i+1}}$, or the empty set if no decrementing path was found. Then, $E_{i+1} = E_i \cup \{v_{i+1}^{out}\} \setminus A_{i+1}$.
\end{obs}
\begin{proof}
    If no decrementing path was found the observation easily follows. On the other hand, if a decrementing path $D$ is found, the observation follows from the fact that the only edge in $D$ of the form $(u^{out}, t)$ with $u^{out} \in E_i$, comes from $A_{i+1}$.
\end{proof}

\begin{obs}\label{obs:level-path-changes}
If $l$ is the smallest level visited by the layered traversal in $\residual{\flowG_{i+1}}{\tmp_{i+1}}$, then $|\pathcover_i^{\ge l'}|$ = $|\pathcover_{i+1}^{\ge l'}|$ for every $l' \in \{1, \ldots, \width(G_i)+1\} \setminus \{l+1\}$, and $|\pathcover_{i+1}^{\ge l+1}|=|\pathcover_i^{\ge l+1}| +1$.
\end{obs}

Therefore, this is the only way \invC{} can be broken by the algorithm. As such, after the level and path ID updates, we check if $|\pathcover_{i+1}^{\ge l}| = |\pathcover_{i+1}^{\ge l+1}|$, and in that case we decrease the level of every vertex $u$, $\ell(u) \ge l$, by $1$. If this happens, we say that we \emph{merge} layer $l$. 

The running time of all these updates is bounded by $O(|\pathcover_{i}|) = O(k)$ per vertex of level $l$ or more, which dominates the running time of an step of the algorithm (except the initial sparsification).

\begin{figure}[t]
    \centering
    \begin{subfigure}[b]{0.32\textwidth}
        \centering
        \includegraphics[width=0.95\textwidth]{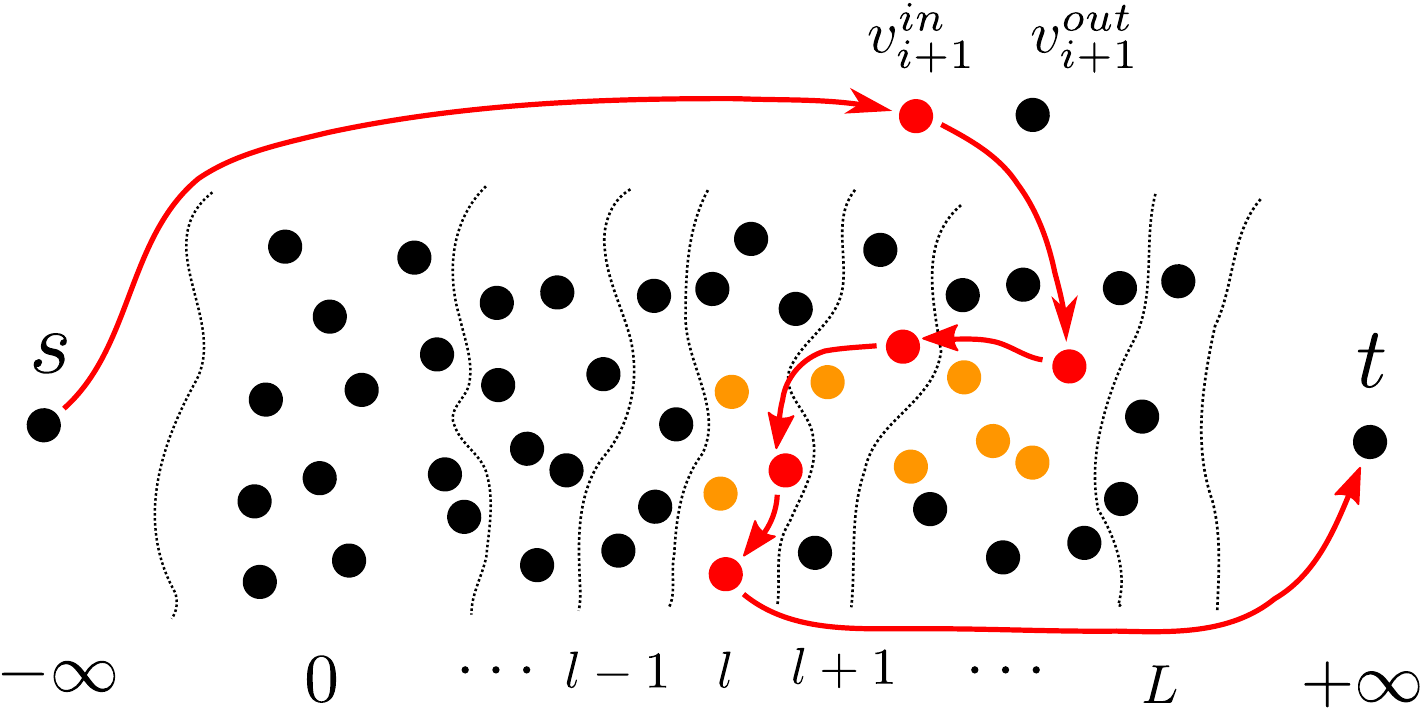}
        \caption[]%
        {{\small Layered traversal}}
        \label{subfig:traversal}
    \end{subfigure}
    \begin{subfigure}[b]{0.32\textwidth}
        \centering
        \includegraphics[width=0.95\textwidth]{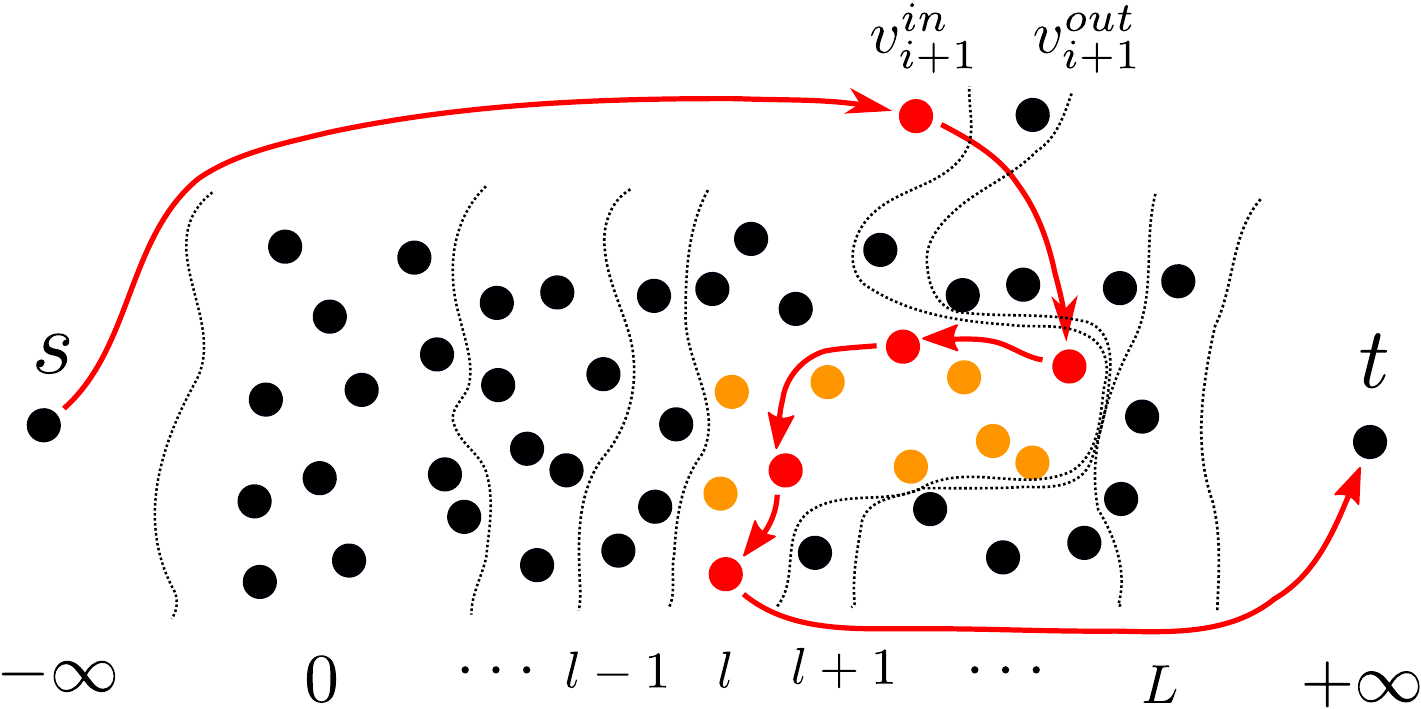}
        \caption[]%
        {{\small Level updates}}
        \label{subfig:updates}
    \end{subfigure}
    \begin{subfigure}[b]{0.32\textwidth}
        \centering
        \includegraphics[width=0.95\textwidth]{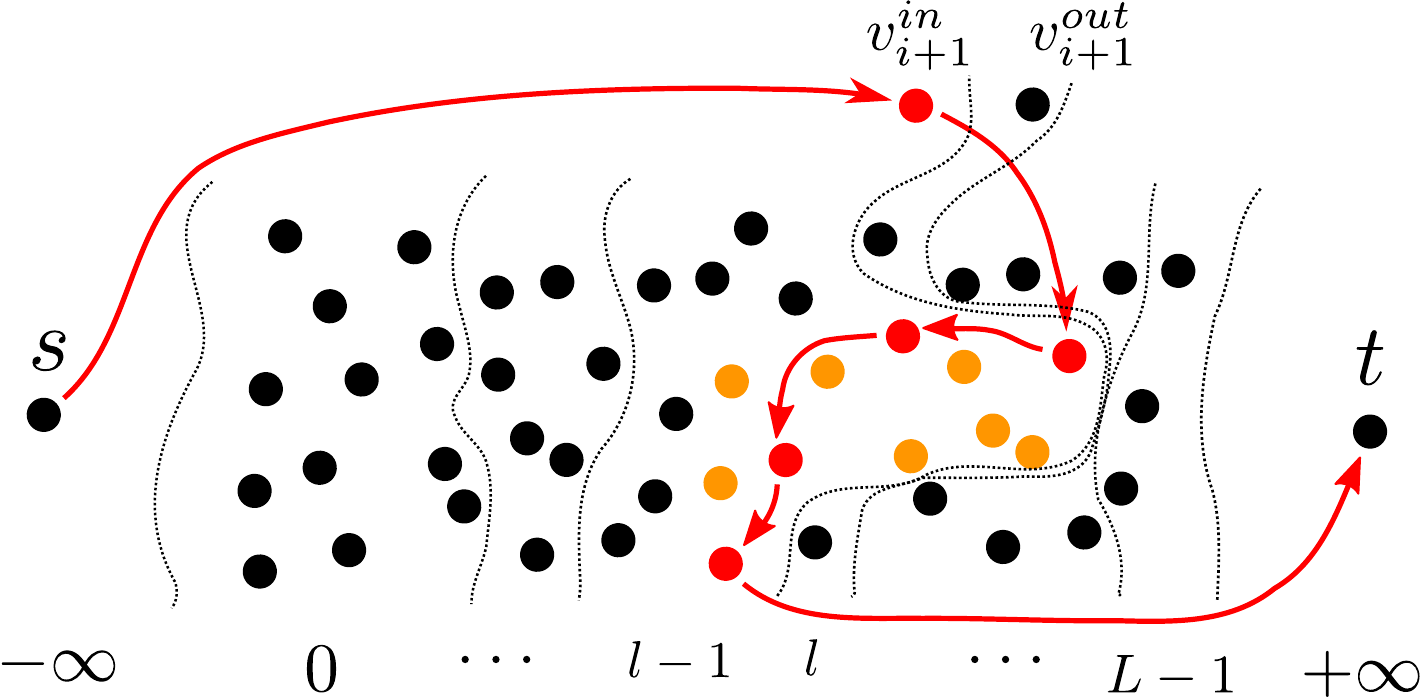}
        \caption[]%
        {{\small Merge of layer $l$}}
        \label{subfig:merge}
    \end{subfigure}
    \caption[]
    {\small Execution of our second algorithm in an abstract example graph. Edges and flow paths are absent for simplicity. Layers are divided by dotted vertical strokes, $L = \max_{v\in\flowV_i\setminus\{t\}}\ell(v)$. \Cref{subfig:traversal} shows a decrementing path in $\residual{\flowG_{i+1}}{\tmp_{i+1}}$ (red) found by the layered traversal as well as all vertices visited (red and orange), $l$ is the smallest layer visited. \Cref{subfig:updates} shows the updates to the level assignment, all vertices visited by the traversal get level $l$, and $v_{i+1}^{out}$ gets level $l+1$. \Cref{subfig:merge} shows the result of merging layer $l$, all vertices of level $l$ or more decrease their level by one.}
    \label{fig:progressive}
\end{figure}

\Cref{fig:progressive} illustrates the evolution of the level assignment in a step of the algorithm.

\subsection{Running time}\label{sec:running-time}
Note that the running time of step $i+1$ is bounded by $O(|N^{-}(v_{i+1})|)$ (from sparsification) plus $O(|\pathcover_{i}|) = O(k)$ per vertex whose level is $l$ or more, where $l$ is the smallest level visited by the layered traversal in  $\residual{\flowG_{i+1}}{\tmp_{i+1}}$. The first part adds up to $O(|E|)$ for the entire algorithm, whereas for the second part we show that every vertex is charged $O(k)$ only $O(k^2)$ times in the entire algorithm, thus adding up to $O(k^3|V|)$ in total.
Every time a vertex $u$ is charged $O(k)$, then the minimum level visited in that step must be $l\le \ell(u)$. Consider the sequence $(|\pathcover_i^{\ge 1}|,|\pathcover_i^{\ge 2}|,\ldots,|\pathcover_i^{\ge \ell(u) + 1}|)$ and its evolution until its final state $(|\pathcover_{|V|}^{\ge 1}|,|\pathcover_{|V|}^{\ge 2}|,\ldots,|\pathcover_{|V|}^{\ge \ell'(u) + 1}|)$ (where $\ell'$ is the level assignment when the algorithm finishes).  By \Cref{obs:level-path-changes}, any update that charges $u$ changes exactly one value in this sequence ($|\pathcover_{i}^{\ge l+1}|$ is incremented by one), and possibly truncates the sequence on the right due to $u$'s level being lowered (levels can only decrease over time).  By \invC, this sequence is always strictly decreasing, and since $|\pathcover_i^{\ge 1}| \le k$, it can be updated at most $O(k^2)$ times until it reaches its final state; hence $u$ is charged $O(k)$ only $O(k^2)$ times.

\subsection{Invariants}\label{sec:invariants}

In this section we show that the invariants of the algorithm (\Cref{sec:levels-and-invariants}) are maintained for the next step, namely that the invariants hold for $\flowG_{i+1}, \pathcover_{i+1}$ and the modified level assignment $\ell$.

\begin{description}
    \item[Invariant A]: Let us consider the residual network after level updates have been made to all visited vertices in $\flowV_{i+1}$ and $v_{i+1}^{out}$, but prior to possibly merging layer $l$ (if called for). Consider an edge $(u,v)$ in $\residual{\flowG_{i+1}}{\pathcover_{i+1}}$ with $\{u, v\} \cap \{s, t\} = \emptyset$.  If both $u$ and $v$ are visited, their levels are each set to $l$, so $\ell(u) \ge \ell(v)$ in $\flowG_{i+1}$.  If neither $u$ nor $v$ are visited, both vertices exist in $\flowV_i$ and the flow between these vertices is not modified by the decrementing path $D$, so $(u,v)$ in $\residual{\flowG_i}{\pathcover_i}$. Thus $\ell(u) \ge \ell(v)$ in $\flowG_{i+1}$ by the invariant of the previous iteration, since their levels are unchanged.  If $u$ is visited and $v$ is not, then again $(u,v)$ cannot belong to $D$, so either $u = v_{i+1}^{in}$, or $(u,v)$ in $\residual{\flowG_i}{\pathcover_i}$. In any case $(u, v)$ in $\residual{\flowG_i}{\tmp_{i+1}}$, thus it must be that $\ell(v) \le l$ in $\flowG_i$, otherwise $v$ would be visited during the layered traversal, so again $\ell(v) \le \ell(u)$ in $\flowG_{i+1}$, once $\ell(u)$ has been updated.  If $v$ is visited and $u$ is not, again $(u,v)$ cannot belong to $D$, so either $(u,v) = (v_{i+1}^{out},v_{i+1}^{in})$ and the invariant is maintained by the level assignment, or $(u,v)$ in $\residual{\flowG_i}{\pathcover_i}$, in which case $l \le \ell(v) \le \ell(u)$ prior to setting $\ell(v) = l$.  Thus, the invariant is maintained in all cases. Finally, it is easy to see that a merge of layer $l$ does not break the invariant.

    \item[Invariant B]: By \Cref{obs:end-vertices} the last edges of the paths of $\pathcover_{i+1}$ are $(v_{i+1}^{in}, v_{i+1}^{out})$, or of the form $(u, v)$ with $v \in E_i\setminus A_{i+1}$. As such, after splicing but before a possible merge of layer $l$, the invariant is maintained because the algorithm sets $\ell(v_{i+1}^{in}) = \ell(v_{i+1}^{out})-1 = l$, and it can only decrease the level of $u$ for the rest of the edges (since the vertices in $E_i\setminus A_{i+1}$ are not visited by the layered traversal). If a merge of layer $l$ happens then both extremes of each edge decrease their level by $1$, thus not breaking the invariant.
    \item[Invariant C]: By \Cref{obs:level-path-changes}, the only possibility to break the invariant is that $|\pathcover_{i+1}^{\ge l}| = |\pathcover_{i+1}^{\ge l+1}|$, but if this happens it is fixed by merging layer $l$.
    
\end{description}

\section{Support Sparsification Algorithm}
\label{sec:edge-thinning}

We present an algorithm that transforms any path cover $\pathcover, |\pathcover| = t$ of a DAG $G = (V, E)$ into one of the same size and using less than $2|V|$ distinct edges, in $O(t^2|V|)$ time (\Cref{thm:edge-thinning}). The main approach consists of splicing paths so that edges are removed from the \emph{support} $E_{\pathcover} = \{e \in P \mid P \in \pathcover\}$. It maintains a path cover $\pathcover', |\pathcover'| = t$ of $G' = (V, E_{\pathcover'})$ (thus also a path cover of $G$). At the beginning we initialize $\pathcover'\gets \pathcover$, and we splice paths so that at the end $|E_{\pathcover'}| < 2|V|$.

To decide how to splice paths, we color the vertices of $v \in V$ based on their degree, that is, if $deg_{G'}(v) \le 2$ we color $v$ \blue{}, and \red{} otherwise. We also color the edges $(u, v) \in E_{\pathcover'}$ according to the color of their endpoints, that is, if both $u$ and $v$ are \blue{}, we color $(u, v)$ \blue{}, likewise if both $u$ and $v$ are \red{}, we color $(u, v)$ \red{}, otherwise we color $(u, v)$ \purple{}. We traverse the underlying undirected graph of $G'$ in search of a \red{} cycle  (cycle of \red{} edges) $C$ and splice paths along $C$ so that at least one \red{} edge is removed from $G'$. We repeat this until no \red{} cycles remain, thus at the end we have that \red{} vertices and edges form a forest, \blue{} vertices and edges form a collection of vertex-disjoint paths and cycles, and \purple{} edges connect \red{} vertices with the extreme vertices of \blue{} paths. As such, if the number of \blue{} and \red{} vertices is $n_b$ and $n_r$, respectively, and the number of \blue{} paths is $p$, there are $n_b - p$ \blue{} edges, less than $n_r$ \red{} edges, and at most $2p$ \purple{} edges. Therefore, $|E_{\pathcover'}| < n_b -p + n_r + 2p \le 2|V|$, as desired. The following remark shows that the factor $2$ from the bound is asymptotically tight.

\begin{figure}[t]
    \centering
    \includegraphics[scale=0.4]{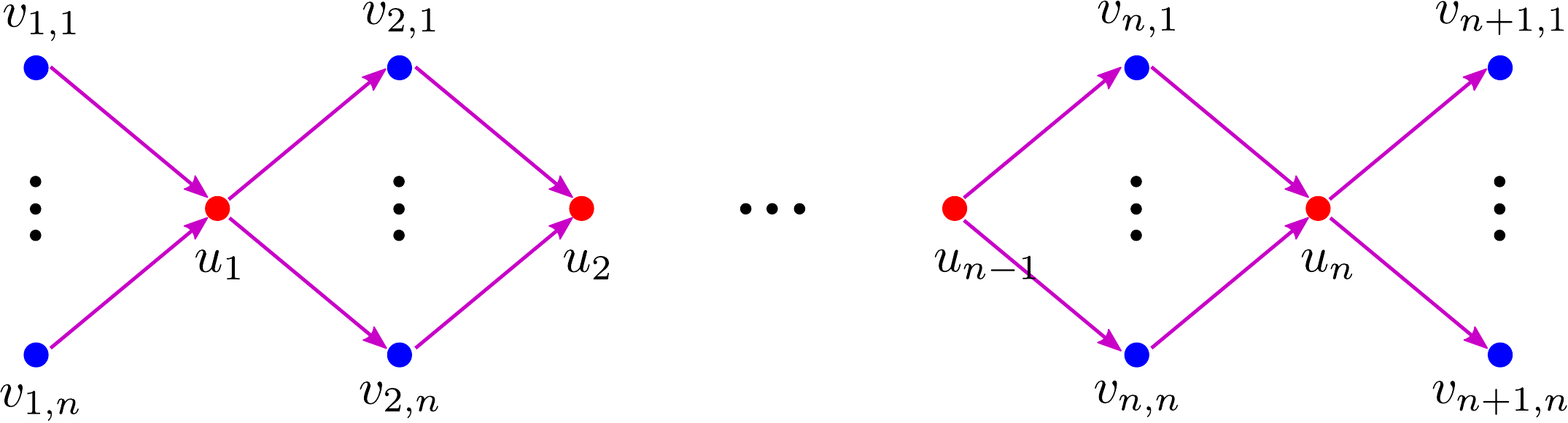}
    \caption{\small A DAG $G$ showing that the factor $2$ from the bound of \Cref{thm:edge-thinning} is asymptotically tight. The figure shows the example graph, as well as the result of applying \Cref{thm:edge-thinning} on an MPC of it. The algorithm colors vertices $v_{i,j}$ \blue{}, vertices $u_i$ \red{}, and edges \purple{}, thus it does not find any \red{} cycle.}
    \label{fig:2-tight}
\end{figure}

\begin{remark}\label{remark:2-tight}
Consider the DAG $G = (V,E)$ from \Cref{fig:2-tight}, with $|V| = n+ n(n+1) = n(n+2)$, $|E| = 2n^2$ and $\width(G) = n$. Note that any path cover $\pathcover$ of size $n$ must use every edge of the graph, then $|E_{\pathcover}|/|V| = |E_{\pathcover'}|/|V| = 2 - 4/(n+2)$.
\end{remark}

Recall that the \emph{multiplicity} of an edge $e$, $\mu(e)$, is the number of paths in $\pathcover'$ using $e$, that is, $\mu(e) = |\{P \in \pathcover' \mid e \in P\}|$. When processing a \red{} cycle $C = v_1, \ldots, v_l, v_{l+1} = v_1$, we partition the corresponding edges of $G'$ in either \emph{forward} $F = \{(v_{i}, v_{i+1}) \in E_{\pathcover'}\}$ or \emph{backward} $B = \{(v_{i+1}, v_{i}) \in E_{\pathcover'}\}$ edges. We splice either along forward or backward edges depending on the comparison between $\sum_{e \in F} \mu(e)$ and $\sum_{e \in B} \mu(e)$. If $\sum_{e \in F} \mu(e) \ge \sum_{e \in B} \mu(e)$, we only splice along backward edges, otherwise only along forward edges. Here we only describe the former case, the later is analogous.
The splicing procedure considers the backward \emph{segments} of the cycle, namely, maximal subpaths of consecutive backward edges in $C$. For each backward segment $b$, it generates a path $P_b \in P'$ that traverses $b$ entirely, by splicing paths along $b$. For this we apply the splicing procedure of \Cref{result:splicing} on every backward segment, which runs in total time $O(|B|) = O(|C|)$. After that, for every $P_b$ we remove $b$ and reconnect the parts of $P_b$ entering and exiting $b$ to their corresponding adjacent forward segments. Note that vertices of $b$ are still covered by some path after splicing since they are \red{}, and the splicing procedure preserves the multiplicity of edges. Also note that the net effect is that the number of paths remains unchanged but the multiplicity of forward edges has increased by one and the multiplicity of backward edges has decreased by one, thus the condition $\sum_{e \in F} \mu(e) \ge \sum_{e \in B} \mu(e)$ will be valid again after the procedure. As such, we repeat the splicing procedure until some backward edge has multiplicity $0$, removing $C$ in this way.

To analyze the running time of all splicing procedures during the algorithm, we consider the function $\Phi(G') = \sum_{e \in E_{\pathcover'}} \mu(e)^2$. We study the change of $\Phi(G')$ of applying the splicing procedure, $\Delta \Phi$. Since the only changes on multiplicity occur on forward and backward edges we have that 
\begin{align*}
\Delta \Phi &= \sum_{e\in F} \left((\mu(e)+1)^2-\mu(e)^2\right) + \sum_{e\in B} \left((\mu(e)-1)^2-\mu(e)^2\right)\\
&= |F| + |B| + 2 \left(\sum_{e \in F} \mu(e) - \sum_{e \in B} \mu(e)\right) \ge |C|.
\end{align*}

As such, each splicing procedure takes $O(|C|)$ time, and increases $\Phi(G')$ by at least $|C|$. Since at the end $\Phi(G') \le  t^2|E_{\pathcover'}|  \le t^22|V|$, the running time of all splicing procedures amounts to $O(t^2|V|)$.

Finally, we describe how to traverse the underlying undirected graph of $G'$ while detecting \red{} cycles in linear time, which is $O(t|V|)$. We perform a \emph{modified} DFS traversal of the graph. We additionally mark the edges as \emph{\processed{}} either when the edge is removed (gets multiplicity $0$), or when the traversal \emph{pops} this edge from the DFS stack\footnote{When an edge is marked as \processed{} we move it at the end of the adjacency list of the corresponding vertex. Therefore, the first edge in the adjacency list of a vertex is always not marked as \processed{}, unless all of them are.}. Since our graph is undirected, all edges are between a vertex and some ancestor in the DFS tree (no crossing edges), thus cycles can be detected by checking if the vertex being visited already is in the DFS stack (and it is not the top of the stack)\footnote{We can maintain an array \emph{in-stack} indicating whether a vertex is in the DFS stack.}. When a \red{} cycle is detected, then we pop from the DFS stack all vertices of the cycle, but without marking as \processed{} the corresponding edges. The cost of these pops plus the additional cost of traversing the edges of the cycle again in a future traversal is linear in the length of the cycle, thus these are charged to the corresponding splicing procedures of this cycle, and the cost of the traversal remains proportional to the size of the graph. 



\bibliographystyle{unsrt}
\bibliography{references}

\appendix
\section{Proof of \Cref{thm:parallel-d-c}}\label{sec:omitted-proofs}

\parallelsparsealgorithm*

\begin{proof}
    We use our algorithm from \Cref{thm:main-d-c}. Since the algorithm divides the problem into two disjoint subgraphs we can easily solve each sub-part by using separate processors, and then join the solutions in $O(|E| + k^2|V|)$ parallel steps. We first subdivide the problem into $O(\log{|V|})$ separate processors, that is, when the size of the input is $O(|V|/\log{|V|})$. We then we run the algorithm in the $O(\log{|V|})$ inputs in parallel, running in $O(|E| + k^2\left(|V|/\log{|V|}\right)\log{\left(|V|/\log{|V|}\right)}) = O(|E| + k^2|V|)$ parallel steps. Finally, all merges (sparsifying and shrinking) from the $O(\log{|V|})$ processors up to the root of the recursion tree are performed level by level. We execute the merges of a level in parallel, thus adding up to $O(|E| + k^2|V|)$ parallel steps in total.
\end{proof}

\section{Full version of \Cref{sec:preliminaries}} \label{sec:extended-preliminaries}

This section is a full version of \Cref{sec:preliminaries}. For the sake of completeness, since there is no definitive reference for some of these notions and results (with some of them considered folklore), we include their full definitions and proofs here.

\subsection{Basics}\label{sec:basics-extended}
A directed graph is a tuple $G = (V, E)$, where $V$ is a set of vertices and $E$ is a set of edges, $E\subseteq V^2$. For an edge $e=(u,v) \in E$, it is said that $e$ goes \emph{from} $u$ \emph{to} $v$, that $u$ and $v$ are \emph{neighbors}, and that $e$ is \emph{incident} to both $u$ and $v$. In particular, $u$ is an \emph{in-neighbor} of $v$, $v$ is an \emph{out-neighbor} of $u$, $e$ is an edge \emph{incoming} to $v$ and \emph{outgoing} from $u$. We denote $N^+(v)$ ($N^-(v)$) to the set of out-neighbors (in-neighbors) of $v$, and by $I^+(v)$ ($I^-(v)$) the edges outgoing (incoming) from (to) $v$. A graph $S = (V_S, E_S)$ is said to be a \emph{subgraph} of $G$ if $V_S \subseteq V$ and $E_S \subseteq E$. If $V_S = V$ it is called \emph{spanning} subgraph. If $V' \subseteq V$, then $G[V']$ is the subgraph of $G$ \emph{induced} by $V'$, defined as $G[V'] = (V', E_{V'})$, where $E_{V'} = \{(u, v) \in E ~:~ u,v \in V'\}$. A \emph{path} $P$ in $G$ is a sequence of vertices $v_1, \ldots , v_{\ell}$ of $G$, such that $(v_i, v_{i+1}) \in E$, for all $i \in [1\ldots \ell-1]$, and $v_{i}\not= v_{j}$, for all $i\not=j$. For every $i,j \in\{1,\ldots,\ell\}, i\le j,$ $v_{i},\ldots,v_j$ is a \emph{subpath} of $P$. If $v_{1} = v_{\ell}$ it is called \emph{cycle}, and we denote it by $C$. If $\ell\ge 2$ it is said that the path is \emph{proper}. A \emph{directed acyclic graph} (DAG) is a directed graph without proper cycles. A \emph{topological ordering} of a DAG is a total order of $V$, $v_1,\ldots, v_{|V|}$, such that for all $(v_{i}, v_{j}) \in E$, $i < j$. A topological ordering can be computed in $O(|V|+|E|)$ time~\cite{kahn1962topological,tarjan1976edge}. If there exists a path $P = v_1, \ldots , v_{\ell}$ in $G$, with $u = v_1$ and $v = v_{\ell}$, it is said that $u$ \emph{reaches} $v$. A \emph{path cover} $\pathcover$ of $G$ is a set of paths such that every vertex $v\in V$ appears in some path of $\pathcover$. If $\pathcover$ has maximum size among all path covers, then it is a \emph{minimum path cover} (MPC), and its size corresponds to the \emph{width} of $G$, that is, $\width(G) = \min_{\pathcover, \text{path cover}} |\pathcover|$ . An \emph{antichain} $A$ is a set of vertices such that for each $u,v \in A$ $u\not=v$ $u$, does not reach $v$, a \emph{maximum antichain} is an antichain of maximum size. Dilworth's theorem~\cite{dilworth2009decomposition} states that the size of a maximum antichain equals the size of an MPC. The \emph{multiplicity} of an edge $e\in E$ with respect to a set of paths $\pathcover$, $\mu_{\pathcover}(e)$ (only $\mu(e)$ if $\pathcover$ is clear from the context), is defined as the number of paths in $\pathcover$ that contain $e$, $\mu_{\pathcover}(e) = |\{P\in\pathcover\mid e \in P\}|$.

In our algorithms we work with subgraphs induced by a consecutive subsequence of vertices in a topological ordering. As such, the following lemma, proven by Cáceres et.al~\cite{caceres2021a}, shows that we can bound the width of these subgraphs by $k = \width(G)$.

\topologicalDoNotIncreaseWidth*

\subsection{Minimum Flow}\label{sec:min-flow-extended}
The problem of minimum flow with lower and upper bounds on edges has been studied before (see for example \cite{ahujia1993network,ciurea2004sequential,bang2008digraphs}). The concept of maximum ow-cuts has been studied before but only in the context of some specific problem solved by a reduction to minimum flow (see for example \cite{mohring1985algorithmic,pijls2013another,marchal2018parallel}). For completeness, in this section we include a proof for the case when only lower bounds on the edges are considered. The proof shown is an adaptation of the proof of the maximum flow/minimum cut theorem given in~\cite{williamson2019network}.

Given a (directed) graph $G = (V, E)$, a source $s \in V$, a sink $t \in V$, and a function of \emph{lower bounds} or \emph{demands} on its edges $d: E \rightarrow \Nzero$, an $st$-\emph{flow} (or just \emph{flow} when $s$ and $t$ are clear from the context) is a function on the edges $f: E \rightarrow \Nzero$, satisfying $f(e) \ge d(e)$ for all $e \in E$ ($f$ \emph{satisfies the demands}) and $\sum_{e \in I^-(v)} f(e) = \sum_{e \in I^+(v)} f(e)$ for all $v \in V \setminus \{s,t\}$ (\emph{flow conservation}). If a flow exists, the tuple $(G, s, t, d)$ is said to be a \emph{flow network}. The \emph{size} of $f$ is the net amount of flow exiting $s$, formally $|f| = \sum_{e \in I^+(s)} f(e) - \sum_{e \in I^-(s)} f(e)$. An $st$-\emph{cut} (or just \emph{cut} when $s$ and $t$ are clear from the context) is a partition $(S , T)$ of $V$ such that $s \in S$ and $t \in T$. An edge $(u, v)$ \emph{crosses} the cut $(S, T)$ if $u\in S$ and $v \in T$, or vice versa. If there are no edges \emph{crossing} the cut from $T$ to $S$, that is, if $\{(u,v) \in E \mid u \in T, v \in S\} = \emptyset$, then $(S, T)$ is a \emph{one-way cut} (ow-cut). The \emph{demand} of an ow-cut is the sum of the demands of the edges crossing the cut, formally $d((S,T)) = \sum_{e = (u, v), u \in S, v \in T} d(e)$. An ow-cut whose demand is maximum among the demands of all ow-cuts is a \emph{maximum ow-cut}.

From these definitions the following properties can be derived:
\begin{basicprop*}\label{prop:basicsMinFlow}
For a flow network $(G, s, t, d)$:
\begin{enumerate}
    \item[(a)] For any cut $(S,T)$ and flow $f$:
    \begin{align*}
        |f| = \sum_{e = (u, v) \in E, u \in S, v \in T} f(e) - \sum_{e = (v, u) \in E, u \in S, v \in T} f(e).
    \end{align*}
    \item[(b)] For any ow-cut $(S,T)$ and flow $f$, $|f| \ge d((S,T))$.
    
\end{enumerate}
\end{basicprop*}
\begin{proof}
  \begin{enumerate}
    \item[(a)] By definition of size, flow conservation and the fact that $(S, T)$ is a partition of $V$.
    \begin{align*}
        |f| &= \sum_{e \in I^+(s)} f(e) - \sum_{e \in I^-(s)} f(e)\\
        &= \left(\sum_{e \in I^+(s)} f(e) - \sum_{e \in I^-(s)} f(e)\right) + \sum_{u \in S\setminus \{s\}}\underbrace{\left(\sum_{e \in I^+(u)} f(e) - \sum_{e \in I^-(u)} f(e)\right)}_{=0}\\
        &=\sum_{u \in S}\left(\sum_{e \in I^+(u)} f(e) - \sum_{e \in I^-(u)} f(e)\right)\\
        &= \sum_{u \in S}\left(\sum_{e = (u, u') \in E, u' \in S} f(e) + \sum_{e = (u, v) \in E, v \in T} f(e)   - \sum_{e = (u', u) \in E, u' \in S} f(e) - \sum_{e = (v, u) \in E, v \in T} f(e)\right)\\
        &= \sum_{u \in S}\left(\sum_{e = (u, v) \in E, v \in T} f(e)  - \sum_{e = (v, u) \in E, v \in T} f(e)\right) + \underbrace{\sum_{u \in S}\left(\sum_{e = (u, u') \in E, u' \in S} f(e) - \sum_{e = (u', u) \in E, u' \in S} f(e)\right)}_{=0}\\
        &= \sum_{u \in S}\left(\sum_{e = (u, v) \in E, v \in T} f(e)   - \sum_{e = (v, u) \in E, v \in T} f(e)\right)\\
        &= \sum_{e = (u, v) \in E, u \in S, v \in T} f(e) - \sum_{e = (v, u) \in E, u \in S, v \in T} f(e)
    \end{align*}
    
    \item[(b)] By using the previous property, the fact that ow-cuts do not have edges crossing from $T$ to $S$ and the lower bounds on the edges.
    \begin{align*}
        |f| &= \sum_{e = (u, v) \in E, u \in S, v \in T} f(e) - \sum_{e = (v, u) \in E, u \in S, v \in T} f(e) \\
        &= \sum_{e = (u, v) \in E, u \in S, v \in T} f(e)\\
        & \ge \sum_{e = (u, v) \in E, u \in S, v \in T} d(e)\\
        &= d((S,T))
    \end{align*}
    
\end{enumerate}  
\end{proof}

Given a \emph{flow network} $(G, s, t, d)$, the problem of \emph{minimum flow} consists of finding a flow $f^*$ of minimum size $|f^*|$ among the flows of the network, such flow is a \emph{minimum flow}. If a minimum flow exists, then $(G, s, t, d)$ is a \emph{feasible} flow network. The following theorem relates the maximum demand of a ow-cut with the size of a minimum flow~\cite{ahujia1993network,ciurea2004sequential,bang2008digraphs}.

\begin{restatable}{theorem}{maxowcutminflow}\label{thm:maxowcut-minflow}
Let $(G, s, t, d)$ be a feasible flow network. Then,
\begin{align*}
    \max_{(S, T), st\text{-}ow\text{-}cut} d((S,T)) = \min_{f, st\text{-}flow} |f|.
\end{align*}
\end{restatable}
\begin{proof}
Given a flow $f$ in $(G, s, t, d)$, the \emph{residual network} of $G$ with respect to $f$ is defined as $\residual{G}{f} = (V, E_f)$ with $E_f = \{(u, v) \mid (v, u) \in E\} \cup \{e \mid f(e) > d(e)\}$, that is, the \emph{reverse edges} of $G$, plus the edges of $G$ on which the flow can be decreased without violating the demands (\emph{direct edges}). Note that a path from $s$ to $t$ in $\residual{G}{f}$ can be used to create another flow $f'$ of smaller size by increasing flow on reverse edges and decreasing flow on direct edges of the path, such path its is called \emph{decrementing path}. Therefore, for a minimum flow $f^*$ there is no decrementing path in $\residual{G}{f^*}$. Taking $S$ as the vertices reachable from $s$ in $\residual{G}{f^*}$ (and $T$ its complement), $(S,T)$ is an ow-cut ($s \in S$, $t \in T$, and there is no edge in $G$ from $T$ to $S$, since there is no edge in the opposite direction in $\residual{G}{f^*}$ by definition of $S$). Moreover, for every edge $e\in E$ from $S$ to $T$, $f(e) = d(e)$, since otherwise this edge would appear in $\residual{G}{f^*}$, which is not possible by definition of $S$. Therefore, the inequality of Property (b) is an equality and $|f^*| = d((S,T))$. Finally, since the demand of any ow-cut is a lower bound for the size of the flow, $(S,T)$ is maximum ow-cut.
\end{proof}



\subsection{MPC in DAGs through Minimum Flow}\label{sec:minflow-reduction-extended}
The reduction from MPC in DAGs to minimum flow has been stated several times in the literature~\cite{ntafos1979path,mohring1985algorithmic,gavril1987algorithms,Jagadish90,ciurea2004sequential,rademaker2012optimal,pijls2013another,marchal2018parallel,makinen2019sparse}, we include it here for completeness.

The problem of finding an MPC in a DAG $G = (V, E)$ can be solved by a reduction to the problem of minimum flow on an appropriate feasible flow network $(\flowG = (\flowV, \flowE), s, t, d)$, defined as: $\flowV = \{s, t\} \cup \{v^{in} \mid v \in V\} \cup \{v^{out} \mid v \in V\}$ ($\{s,t\}\cap V = \emptyset$), that is, the source $s$, the sink $t$ and two vertices $v^{in}, v^{out}$ representing a division of every vertex $v \in V$; $\flowE = \{(s, v^{in}) \mid v \in V\} \cup \{(v^{out}, t) \mid v \in V\} \cup \{(v^{in}, v^{out}) \mid v \in V\} \cup \{(u^{out}, v^{in}) \mid (u,v) \in E\}$, that is, $s$ is connected to all vertices $v^{in}$, $t$ from all vertices $v^{out}$, the split vertices are connected from $v^{in}$ to $v^{out}$ if $v \in V$, and also the topology of $G$ is represented by connecting from $u^{out}$ to $v^{in}$ if $(u,v) \in E$. The demands are defined as $d(e) = 1$ if $e = (v^{in}, v^{out})$ for some $v \in V$ and $0$ otherwise. The tuple $(\flowG, s, t, d)$ is the \emph{flow reduction} of $G$. Note that $|\flowV| = O(|V|)$, $|\flowE|=O(|E|)$, and $\flowG$ is a DAG.

A path cover $\pathcover = P_1, \ldots, P_\ell$ of $G$ directly translates into a flow $f$ for $\flowG, s, t, d$ of size $|f| = \ell$. Starting with a function $f(e) = 0, e \in \flowE$ and iteratively increasing it. For every path $P_i$, it suffices to attach $s$ and $t$ at the ends and to replace every $v \in P_i$ by $v^{in}, v^{out}$, then the flow through the edges of the resulting path is increased by $1$. Since the flow is increased through paths from $s$ to $t$ this procedure maintains the flow conservation constrains, furthermore, since $\pathcover$ is a path cover, the flow through every edge $(v^{in}, v^{out})$ is increased by at least $1$ for every $v \in V$, thus $f$ corresponds to a flow of size $|\pathcover|$.

Moreover, every flow $f$ of $(\flowG, s, t, d)$ can be decomposed into $|f|$ paths corresponding to a path cover of $G$. Iteratively, starting from $f$, a path $P$ from $s$ to $t$ whose edges have positive flow is found, and then the flow on the edges of $P$ is decreased by $1$. By flow conservation, $P$ can be found while $|f| > 0$, and since $|f|$ is decreased by $1$ at each iteration, exactly $|f|$ paths are obtained. By construction of $\flowG$ these paths can easily be transformed into a path cover of size $\ell$ of $G$, by removing $s$ and $t$ and merging the split vertices.

As such, a minimum flow of $(\flowG, s, t, d)$ provides an MPC of $G$. Moreover, the set of edges of the form $(v^{in}, v^{out})$ crossing a maximum ow-cut corresponds to a maximum antichain of $G$ (by merging the edges $(v^{in}, v^{out})$ into $v$, see~\Cref{sec:antichain-structure}). By further noting that if $f$ is a minimum flow of $(\flowG, s, t, d)$, and defining $S = \{v \in \flowV \mid s\text{ reaches } v\text{ in }\residual{\flowG}{f}\}$, then $(S, T=\flowV\setminus S)$ corresponds to a maximum ow-cut, we obtain the following result.
\fastantichain*
\begin{proof}
    We build the flow reduction $(\flowG, s, t, d)$ of $G$, and $\residual{\flowG}{\pathcover}$ in time $O(k|V|+|E|)$. Then, we traverse $\residual{\flowG}{\pathcover}$ and find all vertices reachable from $s$, those vertices form a maximum ow-cut, the edges of the form $(v^{in}, v^{out})$ crossing this cut represent a maximum antichain of $G$ (see \Cref{sec:antichain-structure}).
\end{proof}

\subsection{Sparsification, shrinking, splicing} 

\paragraph{Transitive sparsification.} We say that a spanning subgraph $S = (V, E_S)$ of a DAG $G = (V, E)$ is a \emph{transitive sparsification} of $G$, if for every $u, v \in V$, $u$ reaches $v$ in $S$ if and only if $u$ reaches $v$ in $G$. Since $G$ and $S$ have the same reachability relations on their vertices, they share their antichains, thus $\width(G) = \width(S)$. As such, an MPC of $S$ is also an MPC of $G$, thus the edges $E\setminus E_S$ can be safely removed for the purpose of computing an MPC of $G$. If we have a path cover $\pathcover$ of size $t$ of $G$, then we can \emph{sparsify} (remove some transitive edges) the incoming edges of a particular vertex $v$ to at most $t$ in time $O(t+|N^-(v)|)$. If $v$ has more than $t$ in-neighbors then two of them belong to the same path, and we can remove the edge from the in-neighbor appearing first in the path. We create an array of $t$ elements initialized as $survivor = (v_{-\infty})^t$, where $v_{-\infty} \not \in V$ is before every $v\in V$ in topological order. Then, we process the edges $(u,v)$ incoming to $v$, we set $i \gets path(u)$ ($path(u)$ gives the ID of some path of $\pathcover$ containing $u$) and if $survivor[i]$ is before $u$ in topological order we replace it $survivor[i] \gets u$. Finally, the edges in the sparsification are $\{(survivor[i], v) \mid i\in\{1,\ldots,t\} \land survivor[i] \not = v_{-\infty}\}$.

\obssparsification*

By first computing a $path$ function, and then applying \Cref{obs:sparsificationVertex} to every vertex we obtain.

\sparsificationalgorithm*
\begin{proof}
    Let $\pathcover = P_1,\ldots,P_t$. First, we traverse each path in time $O(t|V|)$ and compute for every vertex $path(v)$, which is the ID of some path containing $v$. We also initialize $v.survivor[i] = u$ if $(u, v)$ is an edge of path $P_i$ and $v_{-\infty}$ if such edge does not exist ($v_{-\infty}\not \in V$, is before every $v \in V$ in topological order). Then, we process the edges $e = (u, v)$ in time $O(|E|)$, set $i = path(u)$, and if $v.survivor[i]$ is before $u$ in topological order, we set $v.survivor[i] = u$. Finally, $E_S$ will be the edges $(v.survivor[i], v)$ such that $v.survivor[i] \not = v_{-\infty}$, thus there are at most $t|V|$. Note that $S$ contains all the edges in the paths because we initialized $v.survivor[i] = u$ for every edge $(u, v)$ in path $P_i$, and these are not updated during the algorithm, thus $\pathcover$ is also a path cover of $S$. Now we prove that $S$ is a transitive sparsification of $G$. If an edge appears in $S$, is of the form $(u = v.survivor[i], v)$ for some edge $(u, v)$ of $G$, thus $S$ is a subgraph of $G$. Finally, if an edge $(u, v)$ is not considered in $S$ it means that there is an edge $(x, v)$ such that, $u, x \in P_i$ with $u$ before $x$ in $P_i$. Therefore, there is a path from $u$ to $v$ using the corresponding edges of $P_i$ followed by $(x, v)$.
\end{proof}

The following lemma shows that we can locally sparsify a subgraph and apply these changes to the original graph to obtain a transitive sparsification.
\sparsificationOfSubgraph*
\begin{proof}
    Since $S'$ is a transitive sparsification of $S$, $E_{S'}\subseteq E_S$ thus $E\setminus (E_S \setminus E_{S'}) \subseteq E$ and then $G'$ is a subgraph of $G$. Now, suppose by contradiction that $u$ and $v$ are connected in $G$ by a path $P$, but they are not connected in $G'$. Then, $P$ contains an edge $e = (a, b)\in E_S\setminus E_{S'}$ disconnecting $b$ from $a$ in $S'$, but since $S'$ is a transitive sparsification of $S$, $a$ is connected to $b$ in $S'$, which is a contradiction.
\end{proof}

\paragraph{Shrinking.} As explained before, shrinking is the process of transforming an arbitrary path cover $\pathcover$ into an MPC, and it can be solved by finding $|\pathcover|-\width(G)$ decrementing paths in $\residual{\flowG}{\pathcover}$, and then decomposing the resulting flow into an MPC.  M\"akinen et al.~\cite{makinen2019sparse} apply this idea to shrink a path cover of size $O(k\log{|V|})$. We generalize this approach in the following lemma.

\shrinkingalgorithm*
\begin{proof}
    We build the flow reduction $(\flowG, s, t, d)$ of $G$, and $\residual{\flowG}{\pathcover}$ in time $O(|V|+|E|)$. Then, we shrink the corresponding flow to minimum by finding $t-k$ decrementing paths in $t-k$ traversals of $\flowG$, and finally, we decompose the minimum flow into an MPC in additional $k$ traversals (one per path) of $\flowG$. In total this takes $O(t(|V|+|E|))$ time.
\end{proof}

This lemma is used by our first MPC algorithm. Our second algorithm also uses the concept of shrinking to obtain an MPC, but refines the search for decrementing paths so that it can be amortized to parameterized linear time  (see \Cref{sec:progressive-flows}).

\paragraph{Splicing.} Our last technique consists in reconnecting paths in a path cover $\pathcover$ so that (after reconnecting) at least one of the paths contains as a subpath a certain path $D$, in time $O(|D|)$. We call this process \emph{splicing} of $\pathcover$ through $D$, and it is used to apply the changes required by a decrementing path in our second MPC algorithm (\Cref{sec:splicing-algorithm}), and also to reconnect paths for reducing the number of edges used by an MPC (\Cref{sec:edge-thinning}). Splicing additionally requires that for every edge $e$ of $D$ there is at least one path in $\pathcover$ containing $e$.

\splicingalgorithm*
\begin{proof}
    We process the edges of $D$ one by one, and maintain a path $P$ of the path cover that contains as subpath a prefix of $D$, at the end of the algorithm $P$ will contain the whole $D$ as a subpath as required. We initialize $P$ to be some path of $\pathcover$ containing the first edge of $D$. Then, when processing the next edge $e$ of $D$, we first check if $e$ is the next edge of $\pathcover$, if so we continue to the next edge of $D$. Otherwise, let $P'$ be a path of the path cover containing $e$, then we connect the prefix of $P'$ until $e$ (excluding) with the suffix of $P$ from the edge previous to $e$ in $D$ (excluding), and we also connect the prefix of $P$ until the edge previous to $e$ in $D$ (including) with the suffix of $P'$ from $e$ (including). Note that each these new connections can be made by manipulating pointers in $O(1)$ time, also note that the new set of paths forms a path cover, and the edges of $G$ preserve their multiplicity, as edges in the path cover are never created or removed, only change of path. 
\end{proof}

Because of the last property of $\pathcover'$, the flow induced by $\pathcover$ is the same as the flow induced by $\pathcover'$. As such, if $\pathcover$ is a flow decomposition of a flow $f$, then $\pathcover'$ is also a flow decomposition of $f$.

\section{Structure of antichains}\label{sec:antichain-structure}

Recall that in \Cref{sec:levels-and-invariants} we defined $\pathcover_i^{\ge l} \subseteq \pathcover_i$, as the flow paths whose level is at least $l$, $\pathcover_i^{\ge l} = \{P \in \pathcover_i \mid \ell(P)\ge l\}$. We also define $\flowV_i^{\ge l} \subseteq \flowV_i$ to be the vertices in the paths of $\pathcover_i^{\ge l}$, $\flowV_i^{\ge l} = \{v \in P \mid P \in \pathcover_i^{\ge l}\}$, and $\flowG_i^{\ge l}$ the graph induced by those vertices, $\flowG_i^{\ge l} = \flowG_i[\flowV_i^{\ge l}]$. Note that $\pathcover_i^{\ge l}$ induces a flow in $\flowG_i^{\ge l}$. Finally, we define $S_i^{< l} \subseteq \flowV_i^{\ge l}$ to be those vertices whose level is less than $l$, $S_i^{< l} = \{v \in \flowV_i^{\ge l} \mid \ell(v) < l\}$.

\begin{lemma}\label{lemma:optimality-by-invariant-a}
In $\flowG_i^{\ge l}$, $(S_i^{<l}, \flowV_i^{\ge l}\setminus S_i^{<l})$ is a maximum ow-cut and $\pathcover_i^{\ge l}$ induces a minimum flow.
\end{lemma}
\begin{proof}
By definition of $S_i^{<l}$ and \invA, there are no edges in $\residual{\flowG_i^{\ge l}}{\pathcover_i^{\ge l}}$ exiting $S_i^{<l}$. As such, there cannot be edges crossing from $\flowV_i^{\ge l}\setminus S_i^{<l}$ to $S_i^{<l}$ in $\flowG_i^{\ge l}$ as these imply reverse residual edges, thus is a ow-cut. Moreover, the flow on every edge crossing the cut from $S_i^{<l}$ to $\flowV_i^{\ge l}\setminus S_i^{<l}$ must be exactly the demand of the edge, otherwise it implies a direct residual edge. Therefore, it is a maximum ow-cut and $\pathcover_i^{\ge l}$ induces a minimum flow.
\end{proof}

There exists a close relation between maximum antichains of a DAG $G$ and maximum ow-cuts on its flow reduction $(\flowG, s, t, d)$ (\Cref{sec:minflow-reduction-extended}), which has been studied before (see for example \cite{mohring1985algorithmic,rademaker2012optimal,pijls2013another,marchal2018parallel}). If $A$ is a maximum antichain of $G$, then the cut $(S,T)$, defined by $S = \{u \in \flowV \mid \exists v \in A, u\text{ reaches }v^{in}\}, T = \flowV\setminus S$ is a maximum ow-cut (since its demand is $|A| = \width(G)$, and the size of a minimum flow of the flow reduction is exactly $\width(G)$, see \Cref{sec:minflow-reduction-extended}). Moreover, if $(S, T)$ is a maximum ow-cut, then the edges of the form $(v^{in}, v^{out})$ crossing the cut, form (after merging every edge in the corresponding vertex) a maximum antichain of $G$ (a path between two vertices implies an edge crossing the cut from $T$ to $S$). As such, each of the maximum ow-cuts $(S_i^{<l}, \flowV_i^{\ge l}\setminus S_i^{<l})$ corresponds to a maximum antichain of an induced subgraph of $G$. In particular, \Cref{lemma:optimality-by-invariant-a} implies that our second algorithm implicitly (through the layer assignment) maintains a sequence of size-decreasing antichains.

\end{document}